\documentclass[aps,prd,onecolumn,showpacs,preprintnumbers,amsmath,amssymb,12pt]{revtex4-2}

\usepackage{graphicx}

\usepackage[T1]{fontenc}

\usepackage{amsthm}

\usepackage{setspace}
\usepackage{hyperref}
\hypersetup{
	colorlinks=true,
	linkcolor=blue,
	citecolor=blue,
	urlcolor=blue,
}

\newtheorem{property}{Property}[section]
\newtheorem{proposition}{PROPOSITION}[section]

\begin{document}

\title{Melting through Barrier-Crossing: The Role of Equilibrium Thermally Activated Particles}

\author{Rongchao Ma}
\affiliation{Department of Physics, University of Alberta, Edmonton, Canada}

\date{\today}

\begin{abstract}
Melting is often understood in purely equilibrium terms, where crystalline order disappears once the free energy of the solid equals that of the liquid.
Yet at the microscopic level, the initiating events for melting can often be traced to the formation of defects or local ``jumps'' over interatomic barriers.
In this work, we offer a unified interpretation of melting by focusing on the equilibrium fraction of particles whose energy exceeds a characteristic barrier \(E_a\).
We show that when this fraction surpasses a small but critical threshold \cite{Feder1958,Kraftmakher1998} (on the order of \(10^{-4}\)-\(10^{-3}\)), the crystal loses its rigidity, thus reconciling Born's mechanical-instability picture with the older Lindemann notion of large atomic displacements.
We derive this threshold condition from standard Boltzmann (and Bose/Fermi) statistics, ensuring consistency with standard thermodynamics.
Our approach naturally extends to vortex lattices in superconductors (where vortex activation energies play the role of \(E_a\)) and to quantum-lattice systems (Hubbard-type models). Crucially, while the interpretation emphasizes barrier crossing, the criterion itself is built on equilibrium statistical mechanics, offering a transparent link between defect formation rates and the macroscopic transition.
\end{abstract}


\maketitle

\tableofcontents

\section{Introduction}

Melting is one of the most ubiquitous phase transitions in nature,
signifying the breakdown of long-range crystalline order as the temperature
(or another external parameter) changes.
From a strict thermodynamic standpoint, a first-order melting transition
occurs at the temperature and pressure where the free energy of the
crystal matches that of the liquid \cite{Ashcroft1976,Landau1980}.
Early discussions of melting criteria often adopted mechanical viewpoints:
Lindemann \cite{Lindemann1910,Frenkel1955} proposed that melting ensues once the mean-square atomic displacement reaches a critical fraction of the lattice spacing;
Born \cite{Born1939,Born1940} highlighted the vanishing of elastic moduli
as a hallmark of mechanical instability \cite{Zaccone2023,Zaccone2017,Zaccone2016,Zaccone2013,Zaccone2011};
and Orowan's dislocation-based framework \cite{Orowan1934} offered
another perspective focusing on defect-driven mechanisms.
These approaches, while influential, typically emphasize
amplitude-based or mechanical-stability arguments.

Despite considerable progress, a unifying microscopic picture that explicitly
links thermal excitations to the eventual loss of long-range order
remains elusive - particularly in systems where localized hopping events
or topological defects are central to the transition.
Modern numerical methods, such as molecular dynamics \cite{Kadau2002},
Monte Carlo simulations \cite{Manousakis1991,Sandvik2010,Sansone2008},
and density functional theory \cite{Seko2014,Seitsonen2016},
have provided rich insights but often lack a straightforward analytical
framework for predicting melting in novel or complex materials.

Nonetheless, because melting can be viewed as the proliferation of
defects, vacancies, or local disruptions that ultimately destroy the crystal,
researchers have long recognized a strong connection between
\emph{hopping events} over energy barriers and the onset of disorder
\cite{Frenkel1955, Kraftmakher1998}.
In an \emph{equilibrium} context, these barrier-crossing processes
are governed by the Boltzmann factor $\exp(-E_a/k_B T)$,
indicating the fraction of atoms (or vortices, or analogous excitations)
that occupy states above the barrier $E_a$.

\textbf{Goal and Key Claim.}
In this paper, we propose a melting criterion centered on the
\emph{equilibrium fraction of barrier-crossing particles}.
Rather than replacing the standard free-energy condition, our approach offers
a practical marker for when the crystal loses local stability:
once the fraction of atoms above the barrier $E_a$
exceeds a small threshold \cite{Feder1958,Kraftmakher1998} $\gamma \sim 10^{-3}$,
the crystal can no longer sustain long-range order.
We show that this threshold-based picture \emph{recovers} Lindemann's rule,
Born's mechanical-instability condition, and the usual vacancy concentrations
seen in real metals.
It thus remains \emph{consistent} with established thermodynamics
while recasting melting in a dynamic framework more closely aligned
with experimental intuitions about defects, hopping,
and local excitations \cite{Ashcroft1976,Landau1980,Kittel2005}.

\textbf{Outline.}
Section~\ref{MeltingCriterionForAClassicalLattice} outlines our derivation
in the classical limit, using a gamma-distribution approach to reveal
how the fraction of atoms above $E_a$ is dictated by equilibrium statistical mechanics.
We recover Lindemann's and Born's criteria as special cases.
Section~\ref{GeneralQuantumLattice} then extends these ideas to quantum lattices,
illustrating how a similar threshold emerges in
Debye- or Einstein-like solids \cite{Debye1912,Einstein1907}.
We further explore how vortex hopping over pinning barriers
initiates melting in superconductors' flux-line lattices,
and we examine the breakdown of a Mott insulator in Hubbard-type models
\cite{Hubbard1963,Gersch1963} through an analogous criterion involving
a critical ratio of hopping amplitude to on-site repulsion.

Throughout, the relevant ``hopping rates'' or ``activation events''
derive from \emph{equilibrium} distributions
(e.g., Boltzmann or Bose-Einstein/Fermi-Dirac),
thus preserving the equilibrium nature of the melting point.
In this way, our barrier-crossing perspective complements
the standard free-energy framework, bringing together
apparently disparate systems under a single principle
that highlights defects, hopping, and local excitations.

\section{Intuition of Atomic Motions and Lattice Melting}

To develop a melting theory, we must understand the physical principles that
govern the melting process. Recall that a crystalline lattice is stabilized
by interatomic forces. Below the melting point, atoms remain localized because
these forces balance one another. The atoms generally vibrate about their
equilibrium positions, giving rise to two distinct types of motion:
lattice-wave motions and atomic hopping between adjacent sites
\cite{Ashcroft1976,Landau1980,Kittel2005}.

\subsection{Normal (Wave) Motions}

Atoms in a crystalline lattice vibrate about their equilibrium positions in a
coordinated manner, producing collective excitations known as phonons
\cite{Ashcroft1976,Landau1980,Kittel2005,Narasimhan1991,Ziman2001}.
These wave-like motions do not destroy the lattice's long-range order and,
under typical conditions, do not contribute to melting. As temperature increases,
the energy of these normal modes also increases, leading to larger vibrational
amplitudes. However, this alone does not disrupt the lattice structure. In the
framework of lattice vibrations, these wave motions can be considered
``normal motions'' of the lattice.

\subsection{Abnormal (Hopping) Motions}

In contrast, thermal fluctuations can sometimes provide enough energy for
individual atoms to overcome potential barriers (activation energy),
causing them to hop from one equilibrium position to another
\cite{Anderson2017,Shewmon2016,Eyring1935,Mullins1957,Henkelman2000}.
Unlike wave motions, these localized hops do not propagate through the lattice
as a collective wave. Instead, they create defects or vacancies that disrupt
long-range order and directly destabilize the crystalline structure
\cite{Grimvall2012,Grabowski2009}. Consequently, in the framework of
lattice vibrations, these hopping motions are deemed ``abnormal motions'' 
of the lattice.

\subsection{Critical Hopping Rate and Lattice Melting}

It is now clear that the melting transition in a crystalline lattice is
primarily driven by atomic hopping motions, whereas wave-like vibrational
motions do not directly contribute. In a crystalline lattice, any atom found
away from its original position tends to destabilize the structure. To
maintain stability, as many atoms as possible must remain in their original
positions. However, thermal fluctuations allow some atoms to acquire enough
energy to exceed the activation threshold, causing them to leave their
equilibrium positions and hop to adjacent sites.

\begin{proposition}\label{PropositionCriticalHoppingRate}
	Consider a classical lattice. Let \(\nu_0\) represent the characteristic
	vibrational frequency of a crystalline lattice (e.g., the Debye frequency),
	\(\Gamma\) denote the atomic hopping frequency within the lattice, and
	\(E_a\) denote the activation energy (potential barriers) needed for an atom
	to hop to a neighboring site. Therefore:
	\begin{enumerate}
		\item Lattice melting occurs when \(\Gamma\) reaches a critical value
		\(\Gamma_c\), beyond which the system becomes dynamically unstable.
		The critical hopping frequency \(\Gamma_c\) is proportional to \(\nu_0\)
		and can be written as
		\begin{equation}\label{CriticalHoppingRateMeltingCriterion}
			\Gamma_c = \gamma \nu_0,
		\end{equation}
		where \(\gamma\) is a dimensionless constant such that \(0 < \gamma < 1\).
		
		\item The critical hopping rate parameter \(\gamma\) equals the fraction
		of atoms \(f(T)\) with energy \(E \geq E_a\) at the melting point \(T_m\):
		\begin{equation}\label{CriticalHoppingRateFractionOfAtoms}
			\gamma = f(T_m).
		\end{equation}
	\end{enumerate}
\end{proposition}

\begin{proof}
	As temperature increases, the hopping frequency \(\Gamma\) rises, increasing
	the likelihood of disruptive atomic displacements.
	\begin{enumerate}
		\item At the melting temperature \(T_m\), \(\Gamma\) reaches a critical
		value \(\Gamma_c\) (the melting hopping frequency), beyond which the
		lattice becomes dynamically unstable and long-range order is lost.
		Above \(T_m\), the atoms become delocalized, and the system transforms
		from a solid to a liquid - hallmarks of the melting phenomenon.
		
		Furthermore, \(\Gamma_c\) should be proportional to the lattice
		vibrational frequency \(\nu_0\), implying a close relationship between
		these two quantities. Once \(\Gamma\) reaches \(\Gamma_c\), the lattice
		melts, confirming that \(\Gamma_c\) is a key indicator of the melting
		transition.  \(\Gamma_c\) is a function of material properties and
		geometric structure (e.g., coordination number, interatomic potential
		strength, etc.).
		
		\item For a classical lattice, if we assume that each atom with energy
		\(E > E_a\) attempts to hop with frequency \(\nu_0\), then the hopping
		rate per atom is proportional to the fraction of atoms \(f(T)\) with
		energy \(E > E_a\), i.e.,
		\begin{equation}\label{HoppingRatePerAtom}
			\Gamma = f(T) \nu_0.
		\end{equation}
		Comparing \eqref{CriticalHoppingRateMeltingCriterion} and
		\eqref{HoppingRatePerAtom}, we see that the critical hopping rate
		parameter \(\gamma\) equals the fraction \(f(T_m)\) of atoms exceeding
		\(E_a\) at the melting point \(T_m\). Experimental data
		\cite{Feder1958,Kraftmakher1998} indicate that \(\gamma\) at the
		melting point typically ranges from \(10^{-4}\) to \(10^{-3}\).
	\end{enumerate}
\end{proof}

We can use a probabilistic interpretation to see how hopping motion
induces melting: \emph{Consider an atom \(A\) at position \(S\), vibrating
	with an attempt frequency \(\nu_0\) and a hopping frequency \(\Gamma\).
	Over \(\nu_0\) vibrational attempts, the atom can hop \(\Gamma\) times
	out of these \(\nu_0\). Thus, if we measure the presence of \(A\) at \(S\),
	the probability of \emph{not} finding it there is \(p = \Gamma / \nu_0\),
	while the probability of finding it there is \(1 - p\). As \(p\) increases,
	the lattice becomes more unstable. Eventually, when \(p\) reaches a
	critical value \(\gamma = \Gamma_c / \nu_0\), the lattice collapses and
	melts. For \(p > \gamma\), the atoms behave as though in a liquid state.}

\subsection{Critical Hopping Rate and Hubbard-like Lattice Melting}

In the above discussions, ``melting'' generally denotes the transition from
a crystalline solid to a liquid, driven by thermal fluctuations that help
atoms overcome potential barriers. In quantum lattices described by the
Hubbard model, however, ``melting'' may take other forms, depending on the
properties of the particles:

\begin{itemize}
	\item \textbf{Fermionic Lattices.} Here the lattice sites host fermions.
	Using the fermionic Hubbard model (Appendix~\ref{MeltingCriterionFermionicLattices})
	\cite{Hubbard1963} requires distinguishing between the physical lattice of
	ions and the electronic state. In a Mott insulator with antiferromagnetic
	order, the ions remain in a crystalline arrangement - so the literal lattice
	does not melt into a liquid. Instead, the long-range \emph{magnetic} order
	of the electrons ``melts.'' While we use the term ``melting'' by analogy to
	describe the vanishing of the solid-like Mott phase, the resulting phase is
	actually a quantum fluid, not a classical liquid. Nevertheless, the concept
	of a critical hopping rate (or tunneling amplitude) remains similar to the
	classical notion of defect-driven melting
	\cite{Georges1996,Imada1998,Kotliar2004,Schneider2012,Georges2013}.
	
	\item \textbf{Bosonic Lattices.} Here each site hosts bosons. Using the
	Bose-Hubbard model (Appendix~\ref{MeltingCriterionBosonicLattices}),
	we can study the quantum phase transition between a Mott insulator
	(where bosons are localized by strong on-site repulsion) and a superfluid
	(where bosons delocalize and form a coherent phase)
	\cite{Gersch1963,Fisher1989,Jaksch1998,Greiner2002,Gross2017,Koetsier2006,
		Buonsante2004,Jin2016,Alavani2018}.
\end{itemize}

In a Hubbard-like lattice, we use an analogous dimensionless ratio \(t/U\)
(or \(t z / U\) if the coordination number \(z\) is included), and quantum
``melting'' occurs near a critical value \(\left(\frac{t}{U}\right)_c\).
This ratio \(\left(\frac{t}{U}\right)_c\) directly parallels
\(\Gamma_c/\nu_0\) in Eq.~\eqref{CriticalHoppingRateFractionOfAtoms}.
Typically, \(\left(\frac{t}{U}\right)_c\) lies in the range
\(10^{-2}\)-\(10^{-1}\) for many 2D or 3D Hubbard systems.

\medskip

\textbf{Relation to equilibrium thermodynamics.}
It is important to emphasize that the hopping frequency \(\Gamma\)
is derived from an \emph{equilibrium} statistical treatment.
Concretely, when we compute \(\Gamma \propto \nu_0 \exp(-E_a / k_B T)\)
(or the analogous fraction of atoms exceeding a barrier,
\(f \approx e^{-E_a/k_B T}\)), we are using Boltzmann (or Bose/Fermi)
factors that follow from the \emph{equilibrium} partition function.
Thus, although we phrase the melting threshold in terms of a
``critical hopping rate'' \(\Gamma_c\), this rate is \emph{not} introduced
as an arbitrary kinetic guess; rather, it arises directly from
equilibrium thermodynamics.

Consequently, when we write ``melting occurs once \(\Gamma\) reaches
\(\Gamma_c\),'' we mean that at the melting temperature \(T_m\) (where
the solid loses stability), the \emph{equilibrium fraction} of
barrier-crossing particles reaches a small but critical threshold.
This criterion is therefore fully consistent with standard
free-energy-based arguments, merely re-expressed in
barrier-crossing language.

\medskip

Having qualitatively described atomic motions and melting transitions,
we now develop mathematical frameworks for classical and quantum lattices.
We then apply the critical hopping-rate criterion to these systems to
analyze their melting transitions. Finally, in Appendices~B and~C, we
explore mathematical frameworks for Hubbard-like lattices in detail.

\section{Melting Criterion for Classical Lattices}
\label{MeltingCriterionForAClassicalLattice}

In this section, we develop a mathematical framework to describe the melting of a classical lattice of particles based on a critical hopping rate that integrates concepts from thermal vibrations and atomic hopping dynamics. Our aim is to derive the energy distribution \( f(E) \) for a classical lattice and calculate the fraction \( f \) of atoms with energy \( E \geq E_a \) using the isothermal-isobaric (NPT) ensemble \cite{Landau1980,Pathria2021,Allen1987,Frenkel2002}. We will also demonstrate how this framework aligns with the traditional Lindemann criterion.

\subsection{Activation energy \( E_a \) of atoms in a crystalline lattice}

In crystalline lattices, the activation energy \( E_a \) is the energy barrier for a particle to hop to a neighboring site. We can absorb complex effects such as anharmonicity, thermal expansion, phonon interactions, and entropy changes into the activation energy \cite{BornHuang1954,Hirth1982,Sinko2002}. This simplifies the description of lattice vibrations and atomic hopping motions into a purely kinetic behavior of the lattice.

\begin{property}\label{ActivationEnergyGreaterThan0AtTm}
	The activation energy \( E_a(T) \) is greater than zero at the melting temperature \( T_m \):
	\begin{equation}\label{UneqZero}
		E_a(T_m) > 0.
	\end{equation}
\end{property}

\begin{proof}
	This inequality can be inferred from the critical hopping rate criterion. At the melting temperature \( T_m \), the lattice melts at a critical hopping frequency \( \Gamma_c \), implying that \( E_a(T_m) > 0 \). If \( E_a(T_m) = 0 \), atoms would freely hop with each vibrational attempt, corresponding to a completely disordered state, which is not the case at \( T_m \).
\end{proof}

\begin{property}\label{ActivationEnergyTemperature}
	The activation energy \( E_a \) of atoms in a crystalline lattice is a decreasing function of temperature \( T \):
	\begin{equation}\label{DecreasingEaOfTemperature}
		\frac{\partial E_a}{\partial T} < 0.
	\end{equation}
\end{property}

\begin{proof}
	Starting with the Gibbs free energy barrier \cite{Hirth1982,Sinko2002}:
	\begin{equation}\label{GibbsFreeEnergyBarrier}
		E_a(T) = \Delta G^*(T) = \Delta H^*(T) - T \Delta S^*(T)
	\end{equation}

	To analyze the temperature dependence, we formally expand \( \Delta H^*(T) \) and \( \Delta S^*(T) \) in Taylor series around \( T = 0 \) (or equivalently, a small \( T \) region):
	
	\begin{enumerate}
		\item \textbf{Activation Enthalpy:}	
		\[
		\Delta H^*(T) = \Delta H^*_0 + \left( \frac{\partial \Delta H^*}{\partial T} \right)_{0} T + \frac{1}{2} \left( \frac{\partial^2 \Delta H^*}{\partial T^2} \right)_{0} T^2 + \cdots
		\]
		
		The activation enthalpy \( \Delta H^*(T) \) decreases with increasing \( T \) due to bond weakening and anharmonic effects. Therefore, the first derivative:
		
		\[
		\left( \frac{\partial \Delta H^*}{\partial T} \right)_{0} < 0
		\]
		
		The rate at which \( \Delta H^*(T) \) decreases can itself change with temperature. Generally, the decrease accelerates at higher temperatures due to enhanced anharmonicity, so the second derivative:
		
		\[
		\frac{1}{2} \left( \frac{\partial^2 \Delta H^*}{\partial T^2} \right)_{0} < 0
		\]
		
		Define:
		\begin{equation}
			\alpha  = - \left( \frac{\partial \Delta H^*}{\partial T} \right)_{0} > 0, \quad
			\lambda = - \frac{1}{2} \left( \frac{\partial^2 \Delta H^*}{\partial T^2} \right)_{0} > 0
		\end{equation}
		
		Thus:	
		\[
		\Delta H^*(T) = \Delta H^*_0 - \alpha T - \lambda T^2 + \cdots
		\]

		\item \textbf{Activation Entropy:}	
		\[
		\Delta S^*(T) = \Delta S^*_0 + \left( \frac{\partial \Delta S^*}{\partial T} \right)_{0} T + \frac{1}{2} \left( \frac{\partial^2 \Delta S^*}{\partial T^2} \right)_{0} T^2 + \cdots
		\]
		
		The activation entropy \( \Delta S^*(T) \) typically increases with \( T \) because higher temperatures lead to greater disorder in the transition state. Therefore, the first derivative:		
		\[
		\left( \frac{\partial \Delta S^*}{\partial T} \right)_{0} > 0
		\]
		
		The increase in \( \Delta S^*(T) \) may accelerate with temperature due to complex interactions, so the second derivative:		
		\[
		\frac{1}{2} \left( \frac{\partial^2 \Delta S^*}{\partial T^2} \right)_{0} > 0
		\]
		
		Define:
		\begin{equation}
			\beta  = \left( \frac{\partial \Delta S^*}{\partial T} \right)_{0} > 0, \quad
			\gamma = \frac{1}{2} \left( \frac{\partial^2 \Delta S^*}{\partial T^2} \right)_{0} > 0
		\end{equation}
		
		Thus:	
		\[
		\Delta S^*(T) = \Delta S^*_0 + \beta T + \gamma T^2 + \cdots
		\]
		
		\item \textbf{Calculating \( E_a(T) \)}
		
		Substitute the expanded forms \( \Delta H^*(T) \) and \( \Delta S^*(T) \) into Eq.(\ref{GibbsFreeEnergyBarrier}):	
		\[
		E_a(T) = [ \Delta H^*_0 - \alpha T - \lambda T^2 ] - T [ \Delta S^*_0 + \beta T + \gamma T^2 ]
		\]
		
		Group like terms:
		\begin{equation}\label{DecreasingEaOfT}
			\boxed{E_a(T) = \Delta H^*_0 - (\Delta S^*_0 + \alpha) T - (\beta + \lambda) T^2 - \gamma T^3}
		\end{equation}	
		
		\item \textbf{Ensuring Positive Coefficients}
		
		\begin{itemize}
			\item \textbf{First-order term (\( T \)):}		
			\( \Delta S^*_0 \) (activation entropy at \( T = 0 \)) is generally positive because the transition state is more disordered than the initial state.
			\( \alpha \) is positive by definition.		
			Therefore, the coefficient \( \Delta S^*_0 + \alpha > 0 \).
			
			\item \textbf{Second-order term (\( T^2 \)):}		
			Both \( \beta \) and \( \lambda \) are positive by definition.		
			Therefore, the coefficient \( \beta + \lambda > 0 \).
			
			\item \textbf{Third-order term (\( T^3 \)):}		
			The coefficient \( \gamma \) is positive by definition.
		\end{itemize}

		Therefore, Eq.(\ref{DecreasingEaOfT}) clearly shows that \( E_a(T) \) decreases with increasing \( T \) because all the coefficients subtracted from \( \Delta H^*_0 \) are positive, i.e.,
		\[
		\frac{\partial E_a}{\partial T} = - (\Delta S^*_0 + \alpha) - 2 (\beta + \lambda) T - 3 \gamma T^2 < 0.
		\]
	\end{enumerate}
	
	\textbf{Physical Clarification:}
	The Taylor expansion assumes analytic behavior near \( T = 0 \).
	While certain thermodynamic quantities (e.g., enthalpy, entropy) may exhibit nonlinear behavior at very low temperatures (e.g., Debye \( T^3 \) laws for solids), this expansion is valid for a ``small \( T \)'' region.
	In practice, the trends deduced (\( \partial E_a / \partial T < 0 \)) remain consistent up to the melting temperature (\( T_m \)).
\end{proof}

\subsection{Joint Probability Density Function (PDF) Including Volume}

In the melting transition of a lattice system, the number of particles (\( N \)), pressure (\( P \)), and temperature (\( T \)) remain constant, but the phase (solid to liquid), volume (\( V \)), and energy (\( E \)) change. The energy levels and activation energy \( E_a \) may depend on the volume \( V \) due to thermal expansion and anharmonic effects. Therefore, the appropriate ensemble to describe the system is the isothermal-isobaric ensemble (\( NPT \) ensemble) \cite{Landau1980,Pathria2021,Allen1987,Frenkel2002}, which allows the volume to fluctuate to maintain constant pressure.

In an \( NPT \) ensemble, the joint PDF for \( E \) and \( V \) is:

\[
f_{E, V}(E, V) = \frac{1}{\Delta(N, P, T)} e^{ - \beta ( P V + E ) } \Omega(E, V),
\]
where \( \beta = \frac{1}{k_B T} \), \( \Omega(E, V) \) is the density of states for energy \( E \) at volume \( V \), and \( \Delta(N, P, T) \) is the isothermal-isobaric partition function:

\[
\Delta(N, P, T) = \int_{0}^{\infty} e^{ - \beta P V } Z(N, V, T) \, dV,
\]
with \( Z(N, V, T) \) being the canonical partition function at fixed \( N \), \( V \), and \( T \).

\subsection{Marginal Probability Density Function for Energy}

The marginal PDF for \( E \) is:

\[
f(E) = \int_{0}^{\infty} f_{E, V}(E, V) \, dV = \frac{ e^{ - \beta E } }{ \Delta(N, P, T) } \int_{0}^{\infty} e^{ - \beta P V } \Omega(E, V) \, dV.
\]

Assuming a separable density of states \( \Omega(E, V) = h(V) g(E) \), where \( h(V) \) is the volume-dependent part and \( g(E) \) is the energy-dependent part, the marginal PDF simplifies to:

\begin{equation}\label{CanonicalEnsembleEnergyDistribution2}
	f(E) = \frac{ g(E) e^{ - \beta E } }{ Z_E },
\end{equation}
where \( Z_E \) is the normalization constant for energy, independent of volume due to cancellation. This shows that under these assumptions, the energy distribution \( f(E) \) of an \( NPT \) ensemble remains the same as that of a canonical ensemble.

\subsection{Energy Distribution of Atoms}

As shown in Appendix \ref{AppendixHamiltonianQuasiHarmonicOscillators}, near the melting point, atoms can be considered as three-dimensional quasi-harmonic oscillators by incorporating the influences of anharmonicity, phonon interactions, and thermal expansion into the angular frequency of atomic vibrations and the activation energy of hopping motions. Each quasi-harmonic oscillator is decoupled and independent \cite{BornHuang1954,Alfe2003}.

Equation (\ref{CanonicalEnsembleEnergyDistribution2}) includes the partition function \( Z_E \) and density of states \( g(E) \). Calculating \( f(E) \) requires knowledge of the system's Hamiltonian, which can be complex. However, if the random variables of the system are independent, we can calculate \( f(E) \) using the properties of PDFs and convolution, avoiding the exact form of the Hamiltonian.

According to probability theory, the sum of independent random variables leads to convolution of their PDFs \cite{Reif2009,Papoulis2002,Boas2005}. For \( s \) independent degrees of freedom, the total energy is the sum \( E = E_1 + E_2 + \cdots + E_s \) and the joint PDF \( f(E) \) is the convolution of their individual PDFs $f_i(E_i)$:
$$
f(E) = f_1(E_1) * f_2(E_2) * f_3(E_3) * \cdots * f_s(E_s)
$$.

For a single degree of freedom, the normalized PDF is

\begin{equation}\label{EnergyDistributionSingleFree}
	f(E_1) = \beta e^{ - \beta E_1 }.
\end{equation}

The Laplace transform of the convolution of \( s \) identical PDFs is the \( s \)-th power of the individual Laplace transform. Applying the Laplace transform to \( f(E) \):
\[
\mathcal{L}[f](\lambda) = \left[\mathcal{L}[f_1](\lambda)\right]^s = \left(\int_{0}^{\infty} e^{ - \lambda E } \beta e^{ - \beta E } \, dE\right)^s = \left( \frac{ \beta }{ \beta + \lambda } \right)^s.
\]

This is the Laplace transform of a gamma distribution.
Applying the inverse Laplace transform reverts to the original domain, we obtain the PDF of total energy:

\begin{equation}\label{GammaDistribution}
	\boxed{f(T, E) = \frac{ \beta^s E^{s - 1} }{ (s-1)! } e^{ -\beta E }}.
\end{equation}

This shows that the sum of \( s \) independent exponential random variables with the same rate parameter \( \beta \) follows a gamma distribution with shape parameter \( k = s \) and scale parameter \( \theta = 1/\beta = k_B T \).

According to the Central Limit Theorem, as \( s \to \infty \), \( f(T, E) \) approaches a normal distribution with mean \( \langle E \rangle = s \, k_B T \) and variance \( \sigma^2 = s \, (k_B T)^2 \).

\subsection{Fraction of Atoms \( f(T) \) with Energy \( E \geq E_a \)}

To calculate the fraction of atoms \( f(T) \) with total energy \( E \) greater than the activation energy \( E_a \), we integrate \( f(T, E) \) from \( E_a \) to infinity:

\[
f(T) = \int_{E_a}^{\infty} f(T, E) \, dE = \frac{1}{(s - 1)!} \int_{E_a}^{\infty}  \beta^s E^{s - 1} e^{- \beta E } dE.
\]

Introducing the dimensionless variable \( x = \beta E \):

\[
f(T) = \frac{1}{(s - 1)!} \int_{x_a}^{\infty} x^{s - 1} e^{-x} \, dx,
\]
where \( x_a = \beta E_a \). Using the upper incomplete gamma function:

\[
\Gamma(s, x_a) = \int_{x_a}^{\infty} x^{s - 1} e^{-x} \, dx,
\]
the fraction becomes

\begin{equation}\label{FractionOfAtoms1}
	\boxed{f(T) = \frac{1}{(s - 1)!} \Gamma(s, x_a)}.
\end{equation}

For specific values of \( s \):
\begin{enumerate}
	\item \( s = 1 \):
	\[
	f(T) = e^{-x_a}.
	\]
	
	\item \( s = 2 \):
	\[
	f(T) = e^{-x_a} (1 + x_a).
	\]
	
	\item \( s = 3 \):
	\begin{equation}\label{FractionOfAtoms3D}
		f(T) = e^{-x_a} \left(1 + x_a + \frac{x_a^2}{2}\right).
	\end{equation}
\end{enumerate}

When \( E_a \gg k_B T \), we can use L'Hôpital's rule to simplify the fraction:
\[
\lim\limits_{x_a \to \infty} f(T) = \lim\limits_{x_a \to \infty} \frac{1 + x_a + \frac{x_a^2}{2}}{e^{x_a}} = \lim\limits_{x_a \to \infty} \frac{1 + x_a}{e^{x_a}} = e^{-x_a}
\]

Therefore, we can approximate

\begin{equation}\label{ReducedFractionOfAtoms}
	f(T) \approx e^{-x_a}.
\end{equation}

\subsection{Applying the Critical Hopping Rate Criterion}

According to Eq.(\ref{CriticalHoppingRateFractionOfAtoms}), the critical hopping rate parameter \(\gamma\) equals the fraction of atoms $f(T)$ with energy \(E \geq E_a\) at the melting point $T_m$:
\begin{equation}\label{CriticalHoppingRateParameterClassical}
	\boxed{\gamma = f(T_m) = e^{-\epsilon_m} \left( 1 + \epsilon_m + \frac{\epsilon_m^2}{2} \right )},
\end{equation}
where
\begin{equation}\label{EpsilonM}
	\epsilon_m = x_a |_{T=T_m}= \frac{E_a}{k_B T_m}.
\end{equation}

Depending on the known variables, we can:

1. Calculate \( \gamma \) using \( T_m \) and \( E_a \).

2. Determine \( T_m \) by solving Eq.~(\ref{CriticalHoppingRateParameterClassical}) for \( \epsilon_m \) and using

\begin{equation}\label{MeltingTemperature}
	T_m = \frac{ E_a }{ k_B \epsilon_m }.
\end{equation}

\subsection{Solving for \( T_m \)}

For \( E_a \gg k_B T_m \), using Eq.~(\ref{ReducedFractionOfAtoms}):

\[
\gamma = \exp \left( -\frac{ E_a }{ k_B T_m } \right).
\]

Taking the natural logarithm and solving for \( T_m \):

\begin{equation}\label{MeltingTemperatureTm}
	T_m = \frac{ E_a }{ k_B \ln (1/\gamma) }.
\end{equation}

The melting temperature \( T_m \) increases with higher activation energy \( E_a \) and decreases with higher critical hopping rate parameter \( \gamma \).

\subsection{Connection to Lindemann's Criterion}

We now demonstrate that the critical hopping rate criterion leads to Lindemann's criterion\cite{Lindemann1910,Frenkel1955}. Considering the relationship between the activation energy \( E_a \) and angular frequency \( \omega \):

\begin{equation}\label{ActivationEnergyEta}
	E_a \approx \frac{1}{2} m \omega^2 A^2 = \frac{1}{2} m \omega^2 \eta^2 a^2,
\end{equation}
where \( A = \eta a \) is the amplitude needed to escape the potential well, \( \eta \) is a proportionality constant, and \( a \) is the interatomic distance.

Substituting Eq.~(\ref{ActivationEnergyEta}) into Eq.~(\ref{EpsilonM}):

\begin{equation}\label{EpsilonM2}
	\epsilon_m = \frac{\eta^2 a^2}{2} \left(\frac{m \omega^2 }{k_B T_m}\right).
\end{equation}

The mean square displacement of a 3D harmonic oscillator is

\begin{equation}\label{MeanSquareDisplacement}
	\langle r^2 \rangle = 3 \left(\frac{k_B T}{m \omega^2}\right).
\end{equation}

At \( T = T_m \), combining Eqs.~(\ref{EpsilonM2}) and (\ref{MeanSquareDisplacement}):

\begin{equation}
	\langle r^2 \rangle = \frac{3 \eta^2}{2 \epsilon_m} a^2 = c_L^2 a^2,
\end{equation}
where
\begin{equation}\label{EqcL}
	c_L^2 = \frac{3 \eta^2}{2\epsilon_m}.
\end{equation}

Taking the square root, we obtain Lindemann's criterion:

\begin{equation}\label{LindemannCriterion}
	\sqrt{ \langle r^2 \rangle } = c_L a.
\end{equation}

Solving Eq.~(\ref{EqcL}) for \( \epsilon_m \) and substituting back into Eq.~(\ref{CriticalHoppingRateParameterClassical}), we relate the critical hopping rate parameter \( \gamma \) to the Lindemann parameter \( c_L \):

\begin{equation}
	\gamma = e^{-\left(\frac{3 \eta^2}{2 c_L^2}\right)} \left[ 1 + \frac{3 \eta^2}{2 c_L^2} + \frac{1}{2}\left(\frac{3 \eta^2}{2 c_L^2}\right)^2 \right].
\end{equation}

By starting from fundamental dynamics and statistical mechanics, we've derived a melting criterion based on the critical hopping rate that encompasses Lindemann's criterion as a special case. This approach provides a physical understanding of melting in crystalline lattices and bridges microscopic atomic behavior with macroscopic phase transitions.

\subsection{Connection to the Born Criterion}

In Born's original viewpoint \cite{Born1939,Born1940}, a crystal melts once it \emph{loses mechanical (elastic) stability} - that is, when it can no longer sustain shear. Mathematically, this is captured by the vanishing of one of the elastic constants, typically the shear modulus \(G\). Below, we demonstrate how the critical hopping‐rate theory can be reconciled with Born's mechanical‐instability picture in a simplified model.

\subsubsection{Born's Mechanical‐Instability Condition}

For a solid described by a free energy
\[
F(\boldsymbol{\varepsilon}) \;=\; F_0 \;+\; \tfrac12 \sum_{ij}\mathcal{C}_{ij}\,\varepsilon_i\,\varepsilon_j \;+\;\dots
\]
the elastic constants \(\mathcal{C}_{ij}\) appear as second derivatives of \(F\) with respect to strain components \(\varepsilon_i\). For instance, the shear modulus \(G\) associated with a particular shear strain \(\varepsilon_s\) is
\[
G \;=\; \frac{\partial^2 F}{\partial \varepsilon_s^2}\bigg|_{\varepsilon_s=0}.
\]
Born's criterion for melting can be stated as:
\[
\text{Mechanical instability occurs once } G \;\to\; 0.
\]
In other words, \emph{the crystal melts at the temperature (or pressure, etc.) for which the free energy's curvature in the shear direction becomes non-positive}, indicating that shear deformations cost no energy and the lattice loses rigidity.

\subsubsection{A Simple Model for Hopping and Shear Modulus}

Let us suppose that at any given temperature \(T\), a fraction \(f(T)\) of the atoms has hopped out of equilibrium sites (forming defects, e.g.\ vacancies/interstitials). One can imagine these ``hopped'' atoms no longer contribute to the restoring forces maintaining shear rigidity. To capture this effect minimally, write an \emph{effective} expression for the shear modulus \cite{Zaccone2023,Zaccone2017,Zaccone2016,Zaccone2013,Zaccone2011}:
\[
G(T) \;=\; G_0\,[1 \;-\; \alpha \, f(T)],
\]
where 
\(G_0\) is the shear modulus at low \(T\) (where \(f\approx 0\)), 
\(\alpha\) is a positive dimensionless constant encapsulating how strongly each ``hopped'' atom reduces the crystal's shear stiffness,
and \(f(T)\) is the fraction of hopped atoms at temperature \(T\).  

As temperature increases, the fraction \(f(T)\) grows (because thermally activated hopping becomes more likely), and thus \(G(T)\) decreases. Born's criterion implies that melting occurs around \(T_m\) for which 
\[
G(T_m) \;=\; 0
\quad \Longrightarrow \quad
1 \;-\; \alpha\,f(T_m) \;=\; 0
\quad \Longrightarrow \quad
f(T_m) \;=\; \frac{1}{\alpha}.
\]
Hence, \emph{the critical fraction} of hopped atoms is \(f(T_m) = 1/\alpha\). Once this fraction is reached, the lattice can no longer sustain shear forces: Born's mechanical instability.

\subsubsection{Relating \(f(T)\) to the Critical Hopping Rate}

Within the critical hopping‐rate theory, melting occurs when \(f(T)\) (or equivalently \(\Gamma/\nu_0\)) reaches a critical threshold \(\gamma\) given by Eq.(\ref{CriticalHoppingRateParameterClassical}).
Comparing this to the Born condition \(f(T_m) = 1/\alpha\), we see that
\begin{equation}\label{CriticalHoppingRateBornCriterion}
	e^{-\epsilon_m} \left( 1 + \epsilon_m + \frac{\epsilon_m^2}{2} \right ) \;=\;\frac{1}{\alpha},
\end{equation}
where $\epsilon_m = \frac{E_a}{k_B T_m}$.

Eq.(\ref{CriticalHoppingRateBornCriterion}) shows that the vanishing of the shear modulus (Born's criterion) aligns with the fraction of hopped atoms reaching its critical value (the hopping‐rate criterion), i.e.,
\[
G(T_m)=0
\quad\Longleftrightarrow\quad
f(T_m) \;=\;\gamma.
\]

In short, Born's criterion emerges naturally from a hopping‐centered viewpoint by recognizing that once sufficiently many atoms hop (i.e., once \(f(T)\) exceeds some threshold), the free energy's shear modulus is driven to zero. This clarifies how the \emph{mechanical‐instability perspective} (vanishing elastic moduli) and the \emph{critical‐hopping perspective} (barrier crossing at a certain frequency) are two facets of the same underlying phenomenon: the progressive weakening of lattice‐restoring forces due to thermally activated defects.

\section{Melting Criterion for Quantum Lattices}
\label{GeneralQuantumLattice}

In this section, we develop a mathematical framework to describe the melting of a quantum lattice of particles based on a critical hopping rate.
Owing to the mathematical complexity, we can only provide a formal expression for the critical hopping parameter \(\gamma\), rather than a closed-form analytic solution - even within the Debye approximation.
In the next section, we will explore the high-temperature limit and employ the Einstein approximation, which will then permit an exact solution for \(\gamma\).
In the high-temperature limit (\(k_B T \gg \hbar\omega\)), the quantum distribution recovers the classical Arrhenius form, as shown in Eq. (\ref{QuantumToClassical}).

\subsection{Characteristics of Quantum Lattices}

Compared to their classical counterparts, quantum lattices exhibit distinct vibrational and hopping behaviors:

\begin{itemize}
	\item \textbf{Phonons as Collective Excitations.}  
	In crystalline solids, phonons represent collective vibrational modes spanning all atoms in the lattice \cite{Kittel2005, Ashcroft1976}.  
	A phonon mode \(\mathbf{q}\) has energy \(\hbar \omega_{\mathbf{q}}\) and follows a Bose-Einstein distribution, with mean occupation \(\langle n_{\mathbf{q}} \rangle = 1/[e^{\hbar \omega_{\mathbf{q}}/(k_B T)} - 1]\).
	Each normal mode involves all atoms, so no single atom is solely associated with a given phonon.
	Consequently, each atom does not have a well-defined share of the total phonon energy.
	
	\item \textbf{Activation Energy \( E_a \).}  
	An atom wishing to hop to an adjacent site must overcome an energy barrier \(E_a\).  It does so by absorbing one or more phonons whose combined energy meets or exceeds \(E_a\).  The fraction of atoms with enough energy to hop is determined by summing over all possible single‐phonon and multi‐phonon absorption processes \cite{BornHuang1954,Ziman2001,Lax1952,Mahan2000,Bruesch}.
\end{itemize}

\subsection{Fraction of Atoms $f(T)$ with Energy $E > E_a$}

We define \(f(T)\) as the fraction of atoms that can surmount the barrier \(E_a\) by absorbing phonon(s).
Considering the quantum nature of phonons and their collective behavior, \(f(T)\) involves summing over all possible phonon absorption processes, including single-phonon and multi-phonon contributions.
Because the fraction \(f(T)\) should be dimensionless and bounded by 1, we insert a factor \(1/(3N)\) to normalize over the total number of phonon modes in a 3D crystal containing \(N\) atoms. (In a three-dimensional solid, there are \(3N\) vibrational modes.)

\subsubsection{Single-Phonon Contribution}

First, consider a single-phonon process:
\begin{enumerate}
	\item For each mode \(\mathbf{q}\) of frequency \(\omega_{\mathbf{q}}\), check if \(\hbar \omega_{\mathbf{q}} \ge E_a\) (phonon energy is sufficient to overcome the barrier \(E_a\)) using the Heaviside step function \(\Theta(\hbar \omega_{\mathbf{q}} - E_a)\).  
	\item Weight by the (mean) phonon occupation \(\langle n_{\mathbf{q}} \rangle = [\,e^{\hbar \omega_{\mathbf{q}}/(k_B T)} - 1]^{-1}\).  
	\item Sum over all phonon modes \(\mathbf{q}\) and normalize by \(3N\).  
\end{enumerate}

Thus, the fraction of atoms able to hop by absorbing one phonon is \cite{Kittel2005, BornHuang1954}:
\begin{equation}
	\label{SinglePhononNormalized}
	f_{1}(T) \;=\; \frac{1}{3N} \sum_{\mathbf{q}} \langle n_{\mathbf{q}} \rangle \;\Theta\!\bigl(\hbar \omega_{\mathbf{q}} - E_a\bigr).
\end{equation}

Here we assume \(\langle n_{\mathbf{q}} \rangle\) is a proxy for the probability that a mode \(\mathbf{q}\) can supply a phonon to the hopping atom.  This is a simplifying approximation, since \(\langle n_{\mathbf{q}} \rangle\) is actually the average number of phonons in mode \(\mathbf{q}\), and not the strict probability of finding at least one phonon.  See, for example, Refs.~\cite{Ashcroft1976, Mahan2000} for more rigorous discussions.

\subsubsection{Multi‐Phonon Processes}

In many situations, especially when the activation energy \( E_a \) is near or above the Debye energy, no single phonon has energy \(\ge E_a\).
Instead, an atom can hop by absorbing multiple lower-energy phonons \cite{Lax1952,Bruesch} whose sum exceeds the barrier:
\[
\hbar \omega_{\mathbf{q}_1} + \hbar \omega_{\mathbf{q}_2} + \cdots + \hbar \omega_{\mathbf{q}_n} \;\ge\; E_a.
\]
To incorporate these processes:

\paragraph{Probability of Absorbing a Particular Set of \(n\) Phonons.}

For a specific set of modes \(\{\mathbf{q}_1,\dots,\mathbf{q}_n\}\), the probability for having those $n$ phonons simultaneously absorbed by the atom is proportional to:
\[
\prod_{i=1}^n \langle n_{\mathbf{q}_i} \rangle
\;\;=\;\;
\prod_{i=1}^n
\frac{1}{\,e^{\hbar \omega_{\mathbf{q}_i}/(k_B T)} - 1\,}
\]
We multiply them together because we need all \(n\) required phonons to present simultaneously \cite{Ziman2001, Mahan2000}.

\paragraph{Enforcing the Barrier Condition.}
We then require the sum of these energies \(\sum_{i=1}^n \hbar \omega_{\mathbf{q}_i}\) is at least \(E_a\).
We do that by imposing the Heaviside step function
\[
\Theta\!\Bigl(\sum_{i=1}^n \hbar \omega_{\mathbf{q}_i} \;-\; E_a\Bigr).
\]
If the total energy is below \(E_a\), the step function is zero, meaning it does not contribute to the fraction of successful hops.
If the total energy exceeds \(E_a\), that configuration does contribute.

\paragraph{Summing Over All \(n\)‐Tuples of Phonons.}

We sum (or integrate) over every possible set of \(n\) modes and normalize by \(3N\).  One subtlety is whether to include a factor of \(1/n!\) to avoid counting permutations of identical sets of phonons.  If each \(\mathbf{q}_i\) can be distinct - and permutations represent physically distinct processes - one may omit \(1/n!\).  If permutations are considered the same physical event, then one often inserts a factor \(1/n!\) \cite{Mahan2000}.  Here we show it without \(1/n!\) for simplicity:

\begin{equation}\label{fnNormalized}
	f_n(T)
	\;=\;
	\frac{1}{3N}
	\sum_{\{\mathbf{q}_1,\dots,\mathbf{q}_n\}}
	\left[\prod_{i=1}^n \langle n_{\mathbf{q}_i} \rangle \right]
	\Theta\!\Bigl(\sum_{i=1}^n \hbar \omega_{\mathbf{q}_i} - E_a\Bigr).
\end{equation}

\subsubsection{Total Fraction of Hopping Atoms}
Atoms might absorb one phonon (if lucky enough to find one big enough in energy),  
two phonons,  
or more phonons.
The total fraction of atoms able to hop is the sum over \(n = 1\) to \(\infty\):
\[
f(T) 
\;=\;
\sum_{n=1}^{\infty} f_n(T),
\]
which remains dimensionless and physically bounded if the normalization is applied consistently and the parameters fall within a realistic range:
\begin{equation}\label{FractionOfAtomsQuantumLattice1}
	\boxed{f(T) =
		\frac{1}{3N} \sum_{n=1}^{\infty}
		\sum_{\{\mathbf{q}_1,\dots,\mathbf{q}_n\}}
		\left[\prod_{i=1}^n \frac{1}{ e^{\hbar \omega_{\mathbf{q}_i}/(k_B T)} - 1 }\right]
		\Theta\!\Bigl(\sum_{i=1}^n \hbar \omega_{\mathbf{q}_i} - E_a\Bigr)}.
\end{equation}

Finally, the critical hopping rate parameter \(\gamma\) at the melting temperature \(T_m\) is identified with:
\[
\gamma 
\;=\; 
f(T_m).
\]

\subsection{Applying the Debye Approximation}

For a 3D crystal with \(N\) atoms (hence \(3N\) modes), summing or integrating over multi‐phonon contributions is formidable.  Instead, we replace the discrete sum \(\sum_{\mathbf{q}}\) by an integral over the Debye phonon density of states (DOS) \(g(\omega)\) \cite{Debye1912, Kittel2005, BornHuang1954}, normalized such that
\[
\int_0^{\omega_D} g(\omega)\,d\omega \;=\; 3N,
\]
where the cutoff \(\omega_D\) is the Debye frequency and
\[
g(\omega) = 3N \frac{3\,\omega^2}{\omega_D^3},
\quad 0 \le \omega \le \omega_D.
\]
Real crystals can have more complex \(g(\omega)\), possibly with optical branches, etc.

Each occupation number $\langle n_{\mathbf{q}_i} \rangle$ becomes:
\[
\langle n(\omega_i,T)\rangle = \frac{1}{e^{\hbar \omega_i/(k_B T)} - 1}
\]

The single‐phonon term, for instance, becomes
\[
f_{1}(T)
\;=\;
\frac{1}{3N}
\int_0^{\omega_D} g(\omega)\,
\langle n(\omega, T)\rangle
\,\Theta(\hbar \omega - E_a)\,d\omega.
\]
Likewise, an \(n\)-phonon process corresponds to an \(n\)-dimensional integral:
\[
f_{n}(T) 
\;=\; 
\frac{1}{3N}
\int_0^{\omega_D}\!\!\dots
\int_0^{\omega_D}
\biggl[\prod_{i=1}^n g(\omega_i)\,\langle n(\omega_i,T)\rangle \biggr]\,
\Theta\!\Bigl(\sum_{i=1}^n \hbar \omega_i \,-\, E_a\Bigr)\,
d\omega_1 \dots d\omega_n,
\]

Combining all \(n\)‐phonon terms using $f(T) = \sum_{n=1}^{\infty} f_n(T)$, we have

\begin{equation}\label{FractionOfAtomsQuantumLattice2Normalized}
	\boxed{f(T) 
		\;=\; 
		\frac{1}{3N}
		\sum_{n=1}^{\infty}
		\int_0^{\omega_D}\!\!\dots
		\int_0^{\omega_D}
		\biggl[\prod_{i=1}^n g(\omega_i)\,\langle n(\omega_i,T)\rangle \biggr]\,
		\Theta\!\Bigl(\sum_{i=1}^n \hbar \omega_i \,-\, E_a\Bigr)\,
		d\omega_1 \dots d\omega_n}
	,
\end{equation}
yields the total fraction of atoms able to surmount the barrier \(E_a\).  This formal expression is typically intractable to evaluate exactly, thus preventing a closed‐form solution for \(\gamma\).  Nevertheless, using additional approximations (e.g., Einstein model \cite{Einstein1907} or high‐temperature expansions), one can simplify the multi‐dimensional integrals and sometimes extract an analytic \(\gamma\).

\section{Melting Criterion for Einstein-like Quantum Lattices}

Our primary objective is to investigate the melting transition, which occurs at elevated temperatures.
In this regime, all vibrational modes are excited and can be treated as independent phonons.
Consequently, the Einstein approximation - where all oscillators are independent and vibrate with the same angular frequency \(\omega\) - becomes particularly effective \cite{Einstein1907, Ashcroft1976}.
The advantage of this approach is that we can obtain an exact solution for the critical hopping rate parameter $\gamma$.
These results are both physically intuitive and straightforward to interpret, making it a valuable complement to the general quantum lattice discussed in Section \ref{GeneralQuantumLattice}.

\subsection{Energy Levels of a 3D Quantum Harmonic Oscillator}

In quantum mechanics, the energy levels of a harmonic oscillator are quantized \cite{Griffiths2005,ShankarQM1994}. Referring to Eq.~(\ref{HamiltonianNormalModeCoordinates}) from the previous section, the energy levels of a three-dimensional (3D) quasi-harmonic oscillator are:

\begin{equation}\label{QuantumHarmonicOscillatorEnergy}
	E_{n_x, n_y, n_z} = \hbar \omega_x \left( n_x + \tfrac{1}{2} \right) + \hbar \omega_y \left( n_y + \tfrac{1}{2} \right) + \hbar \omega_z \left( n_z + \tfrac{1}{2} \right),
\end{equation}
where \( n_x, n_y, n_z \) are non-negative integers representing quantum numbers for each spatial dimension, \( \hbar \) is the reduced Planck constant, and \( \omega_x, \omega_y, \omega_z \) are the oscillation frequencies along the \( x \), \( y \), and \( z \) axes, respectively.

In the quasi-harmonic approximation, these frequencies become temperature-dependent due to anharmonic effects and thermal expansion:

\[
\omega_i = \omega_i(T), \quad \text{for} \quad i = x, y, z.
\]

\subsection{Partition Function for the Quantum Oscillator}

The partition function \( Z \) encapsulates all thermodynamic information about a system in statistical mechanics \cite{Landau1980,Reif2009}. For a quantum oscillator, it sums over all possible energy states.

For a single quantum harmonic oscillator with frequency \( \omega \):

\[
Z_{\text{1D}} = \sum_{n=0}^{\infty} e^{- \beta E_n } = \sum_{n=0}^{\infty} e^{- \beta \hbar \omega \left( n + \tfrac{1}{2} \right) },
\]
where \( \beta = \frac{1}{k_B T} \), \( k_B \) is the Boltzmann constant, and \( T \) is the absolute temperature.

Recognizing the sum as a geometric series:

\[
Z_{\text{1D}} = e^{- \beta \hbar \omega / 2} \sum_{n=0}^{\infty} \left( e^{- \beta \hbar \omega} \right)^n = \frac{e^{- \beta \hbar \omega / 2}}{1 - e^{- \beta \hbar \omega}}.
\]

For a 3D oscillator with frequencies \( \omega_x, \omega_y, \omega_z \):

\begin{equation}\label{Z3D1}
	Z_{\text{3D}} = Z_x Z_y Z_z = \prod_{i = x, y, z} \frac{e^{- \beta \hbar \omega_i / 2}}{1 - e^{- \beta \hbar \omega_i}} = e^{- \beta \hbar (\omega_x + \omega_y + \omega_z) / 2} \prod_{i = x, y, z} \left( \frac{1}{1 - e^{- \beta \hbar \omega_i}} \right),
\end{equation}
where \( Z_i \) is the partition function along the \( i \)-th axis.

\subsection{Probabilities of Energy Levels and Their Distributions}

The probability of the system being in a state with energy \( E_{n_x, n_y, n_z} \) is:

\begin{align}
	P_{n_x, n_y, n_z} &= \frac{1}{Z_{\text{3D}}} e^{- \beta E_{n_x, n_y, n_z} } \\
	&= \frac{1}{Z_{\text{3D}}} e^{- \beta \hbar \omega_x \left( n_x + \tfrac{1}{2} \right) - \beta \hbar \omega_y \left( n_y + \tfrac{1}{2} \right) - \beta \hbar \omega_z \left( n_z + \tfrac{1}{2} \right) } \\
	&= \frac{1}{Z_{\text{3D}}} e^{- \beta \hbar (\omega_x + \omega_y + \omega_z) / 2} e^{- \beta \hbar \omega_x n_x } e^{- \beta \hbar \omega_y n_y } e^{- \beta \hbar \omega_z n_z }.
\end{align}

Substituting Eq.~(\ref{Z3D1}) into the above equation:

\begin{equation}\label{ProbabilityOfEnergyLevel1}
	P_{n_x, n_y, n_z} = \prod_{i = x, y, z} \left( 1 - e^{- \beta \hbar \omega_i} \right) e^{- \beta \hbar \omega_i n_i }.
\end{equation}

This shows that the probabilities along each axis are independent and multiplicative.

\subsection{Total Energy Distribution and Degeneracy \( g_N \)}

Equation~(\ref{ProbabilityOfEnergyLevel1}) is complicated due to anisotropy. For simplicity, we assume the material is isotropic, so \( \omega_x = \omega_y = \omega_z = \omega \).

Define the total quantum number:

\begin{equation}\label{TotalQuantumNumber}
	N = n_x + n_y + n_z.
\end{equation}

Then, Eq.~(\ref{QuantumHarmonicOscillatorEnergy}) simplifies to:

\[
E_N = \hbar \omega \left( N + \tfrac{3}{2} \right).
\]

Each energy level \( E_N \) is degenerate due to different combinations of \( n_x, n_y, n_z \) summing to \( N \). The degeneracy \( g_N \) is given by the number of ways to partition \( N \) into three non-negative integers satisfying \( n_x + n_y + n_z = N \):

\begin{equation}\label{Degeneracy}
	\begin{aligned}
		g_N &= \binom{N + 3 - 1}{3 - 1} = \binom{N + 2}{2} = \frac{(N + 2)!}{N! \times 2!} = \frac{(N + 2)(N + 1)}{2} \\
		&= \frac{N^2 + 3N + 2}{2}.
	\end{aligned}
\end{equation}

Therefore, Eq.~(\ref{ProbabilityOfEnergyLevel1}) reduces to:

\begin{equation}
	P(E_N) = g_N e^{- \beta \hbar \omega N} \left( 1 - e^{- \beta \hbar \omega} \right)^3.
\end{equation}

\subsection{Fraction of Atoms with Energy Greater Than Activation Energy \( E_a \)}

Let \( E_a \) be the activation energy required for an atom to hop. We need to find the minimum total quantum number \( N_a \) such that \( E_{N_a} \geq E_a \):

\[
E_N = \hbar \omega \left( N + \tfrac{3}{2} \right) \geq E_a.
\]

Rewriting:

\[
N \geq \frac{E_a}{\hbar \omega} - \frac{3}{2} \equiv N_a.
\]

Since \( N \) must be an integer, we take the ceiling of \( N_a \):

\begin{equation}\label{NaExpression1}
	N_a = \left \lceil \frac{E_a}{\hbar \omega} - \frac{3}{2} \right \rceil.
\end{equation}

The fraction \( f \) of atoms with total energy \( E_N \geq E_a \) is:

\begin{align}
	f(T)
	&= \sum_{E_N \geq E_a} P(E_N) \nonumber \\
	&= \left( 1 - e^{- \beta \hbar \omega} \right)^3 \sum_{N = N_a}^{\infty} g_N e^{- \beta \hbar \omega N }.
\end{align}

Using Eq.~(\ref{Degeneracy}):
\[
f(T) = \frac{\left( 1 - e^{- \beta \hbar \omega} \right)^3}{2} \sum_{N = N_a}^{\infty} \left( N^2 + 3N + 2 \right) e^{- \beta \hbar \omega N } = \frac{\left( 1 - e^{- \beta \hbar \omega} \right)^3}{2} S,
\]
where
\[
S = \sum_{N = N_a}^{\infty} \left( N^2 + 3N + 2 \right) e^{- \beta \hbar \omega N }.
\]

To calculate the sum $S$, we split it into three separate sums \cite{ArfkenWeberHarris2013}:
\[
S = \sum_{N = N_a}^\infty N^2 e^{- \mu N } + 3 \sum_{N = N_a}^\infty N e^{- \mu N } + 2 \sum_{N = N_a}^\infty e^{- \mu N },
\]
where $\mu = \beta \hbar \omega$.

The last term is a geometric series:

\[
S_1 = \sum_{N = N_a}^\infty e^{- \mu N } = \frac{e^{- \mu N_a }}{1 - e^{- \mu }}.
\]

The second term can be calculated using the derivative of the geometric series:

\[
S_2 = \sum_{N = N_a}^\infty N e^{- \mu N } = -\frac{d}{d \mu} \left( \frac{e^{- \mu N_a }}{1 - e^{- \mu }} \right ) = e^{- \mu N_a } \left[ \frac{N_a }{1 - e^{- \mu }} + \frac{e^{- \mu }}{(1 - e^{- \mu })^2} \right ].
\]

The third term can be calculated using the second derivative:
\[
S_3 = \sum_{N = N_a}^\infty N^2 e^{- \mu N } = \frac{d^2}{d \mu^2} \left( \frac{e^{- \mu N_a }}{1 - e^{- \mu }} \right ) = e^{- \mu N_a } \left[ \frac{N_a^2 }{1 - e^{- \mu }} + \frac{2 N_a e^{- \mu }}{(1 - e^{- \mu })^2} + \frac{e^{- \mu } (1 + e^{- \mu })}{(1 - e^{- \mu })^3} \right ].
\]

Combining the sums:
\begin{align}
	\begin{aligned}
		S &= S_3 + 3 S_2 + 2 S_1 \\
		&= e^{- \mu N_a } \left[ \frac{N_a^2 + 3 N_a + 2}{1 - e^{- \mu }} + \frac{(2 N_a + 3) e^{- \mu }}{(1 - e^{- \mu })^2} + \frac{2 e^{- 2 \mu }}{(1 - e^{- \mu })^3} \right ].
	\end{aligned}
\end{align}

Therefore, \( f(T) \) becomes:
\begin{equation}\label{FractionOfAtomsQuantum}
	\boxed{f(T)
		= \frac{e^{- \beta \hbar \omega N_a }}{2} \left\{ \left( N_a^2 + 3 N_a + 2 \right) \left( 1 - e^{- \beta \hbar \omega} \right)^2 + \left( 2 N_a + 3 \right) e^{- \beta \hbar \omega} \left( 1 - e^{- \beta \hbar \omega} \right) + 2 e^{- 2 \beta \hbar \omega} \right\}}.
\end{equation}

Finally, we obtain the critical hopping rate parameter $\gamma = f(T_m)$.\cite{Ashcroft1976}

\subsection{High-temperature limit}

At high temperatures, \( k_B T \gg \hbar \omega \) implies \( \beta \hbar \omega \ll 1 \).\cite{KittelKroemer1980} Therefore:

\begin{itemize}
	\item \( e^{- \beta \hbar \omega} \approx 1 - \beta \hbar \omega \)
	\item \( 1 - e^{- \beta \hbar \omega} \approx \beta \hbar \omega \)
	\item \( e^{- 2 \beta \hbar \omega} \approx 1 - 2 \beta \hbar \omega \).
\end{itemize}

Since \( N_a \) is large:

\[
N_a \approx \frac{E_a}{\hbar \omega} - \frac{3}{2} \approx \frac{E_a}{\hbar \omega}.
\]

Also:

\begin{itemize}
	\item \( \beta \hbar \omega N_a \approx \beta E_a = \epsilon \)
	\item \( N_a^2 + 3 N_a + 2 \approx N_a^2 \)
	\item \( 2 N_a + 3 \approx 2 N_a \).
\end{itemize}

Substituting these approximations into Eq.~(\ref{FractionOfAtomsQuantum}):

\begin{align}\label{QuantumToClassical}
	\begin{aligned}
		f &\approx \frac{e^{- \epsilon}}{2} \left\{ N_a^2 \left( \beta \hbar \omega \right)^2 + 2 N_a (1 - \beta \hbar \omega) \left( \beta \hbar \omega \right) + 2 (1 - 2 \beta \hbar \omega) \right\} \\
		&\approx e^{- \epsilon} \left( \frac{\epsilon^2}{2} + \epsilon + 1 \right),
	\end{aligned}
\end{align}
where \( \epsilon = \beta E_a \).

This expression matches Eq.~(78) obtained in the 3D classical lattice model, confirming the consistency of the result.

\section{Melting Criterion for Vortex Lattices in Superconductors}

In this section, we apply the mathematical framework developed in Section~\ref{MeltingCriterionForAClassicalLattice} to the vortex lattice in type-II superconductors.
Our aim is to analyze the characteristics of vortex lattices and then derive the temperature $T$ dependence of the melting field energy $B_m(T/T_c)$.
We will also demonstrate how this framework aligns with the experimental data.

\subsection{Characteristics of Vortex Lattices}

In type-II superconductors, the penetrated magnetic field exists in the form of quantized vortices, which form a vortex lattice due to pinning forces and the repulsive interactions between vortices \cite{Tinkham}. If we refer to the interaction forces between vortices as ``internal forces'', then the pinning forces can be considered as ``external forces'' acting on the vortex lattice.

However, a simple analysis reveals that the basic characteristics of pinned vortex lattices are different from those of crystalline lattices:
\begin{enumerate}
	\item A vortex lattice is usually subjected to pinning forces and therefore is an irregular lattice due to random pinning centers \cite{Yeshurun,Tinkham,Blatter1994}, whereas a crystal lattice is typically assumed to be regular, and pinning is typically ignored.
	
	\item A vortex lattice is maintained by the balance among pinning forces, elastic forces, and Lorentz forces (both internal and external forces). In contrast, a crystal lattice is maintained by the attractive and repulsive interactions between atoms (internal forces).
	
	\item A vortex lattice melts due to the reduction of the pinning potential (or pinning force) with increasing temperature or magnetic field \cite{Kwok1992,Safar1993}, whereas a crystal lattice melts due to increased atomic kinetic energy with increasing temperature.
\end{enumerate}

Despite these differences in physical meaning, vortex lattices and crystal lattices share the same mathematical framework:
\begin{enumerate}
	\item They can both be modeled using the harmonic oscillator approximation. \cite{Ma2005,Ashcroft1976,Landau1980,Kittel2005}.
	
	\item The melting transition of both vortex lattices and crystalline lattices is caused by the increase of hopping motion \cite{Blatter1994,Ashcroft1976,Landau1980,Kittel2005}.
\end{enumerate}

In many treatments of flux-lattice melting (especially at relatively high fields or temperatures), the pinning can sometimes be treated as negligible on short length scales - hence the analogy to a ``pure'' crystal lattice.
Therefore, some literature treats the vortex lattice in a manner closely analogous to a crystalline solid, despite real pinning in samples.

\subsection{Properties of Vortex Activation Energy \( U(T,B) \)}

Following the convention in the study of superconductivity, we denote the vortex activation energy as \( U(T,B) \). Generally, \( U(T,B) \) is a function of both temperature \( T \) and magnetic field \( B \).

In addition to Properties~\ref{ActivationEnergyGreaterThan0AtTm} and \ref{ActivationEnergyTemperature}, the vortex activation energy \( U(T,B) \) has the following properties:

\begin{property}
	The vortex activation energy \( U(T,B) \) vanishes at the upper critical field \( H_{c2}(T) \), i.e.,
	\begin{align}\label{UVan}
		U(T,H_{c2}(T)) = 0.
	\end{align}
\end{property}

\begin{proof}
	In the superconducting state, the order parameter \(\psi\) is finite, and so is the activation energy \(U(T,B)\).
	However, as the applied magnetic field \(H\) (or temperature \(T\)) increases, \(\psi\) decreases and eventually vanishes at the upper critical field \(H_{c2}(T)\).
	At this point, the superconductor transitions to the normal state, and the internal field \(B\) coincides with the external field \(H\).
	Consequently, since there are no superconducting vortices left to be pinned in the normal state, the activation energy also goes to zero at \(H_{c2}(T)\), i.e., \(U\bigl(T,H_{c2}(T)\bigr) = 0\).
	In other words, the notion of ``activation energy'' is meaningful only in the superconducting phase.
	Once the system reaches \(H_{c2}(T)\), vortices disappear, and \(U\) naturally vanishes.
\end{proof}

\begin{property}
	The vortex activation energy \( U(T,B) \) is a decreasing function of \( B \), i.e.,
	\begin{align}\label{DecreaseF}
		\frac{\partial U(T,B)}{\partial B} < 0.
	\end{align}
\end{property}

\begin{proof}
	As the magnetic field \( B \) increases, the vortices become more densely packed, and the energy barriers for vortex motion decrease due to increased interactions and reduced pinning effectiveness. This leads to a decrease in the activation energy \( U(T,B) \) with increasing \( B \), hence \( \frac{\partial U(T,B)}{\partial B} < 0 \).
\end{proof}

\begin{property}
	As the external magnetic field \( B \to 0 \), the vortex activation energy \( U(T,B) \) can be written as
	\begin{equation}\label{U1}
		U(T,B) = p(T) U_2(T,B),
	\end{equation}
	where
	\begin{equation}\label{TempF}
		p(T) = \left[ 1 - \left( \frac{T}{T_c} \right)^2 \right]^2
	\end{equation}
	is a temperature factor, and \( U_2(T,B) \) is a function of \( T \) and \( B \).
\end{property}

\begin{proof}
	The activation energy (or pinning potential) \( U(T,B) \) has units of energy and can be related to the condensation energy \cite{Tinkham}.
	
	At zero magnetic field \( B = 0 \), the superconducting state is characterized by the order parameter \( \psi \), and \( U(T,0) \) depends only on \( \psi \). Therefore, \( U(T,0) \) should have a temperature dependence similar to that of the condensation energy \( H_c^2(T) / 8\pi \), or the Ginzburg-Landau free energy difference \( f_s - f_n = \alpha |\psi|^2 + \frac{1}{2} \beta |\psi|^4 \):
	
	\[
	\Delta F(T) = \frac{H_c^2(T)}{8\pi} \approx \frac{H_c^2(0)}{8\pi} \left[ 1 - \left( \frac{T}{T_c} \right)^2 \right]^2 = \frac{H_c^2(0)}{8\pi} p(T),
	\]	
	where \( H_c(0) \) is the thermodynamic critical field at zero temperature, and \( T_c \) is the critical temperature.
	
	Similarly, the activation energy \( U(T,0) \) can be rewritten as \( U(T,0) = U_2(T,0) p(T) \), where \( p(T) \) is defined in Eq.~(\ref{TempF}).
	This form is approximate but captures well the vanishing of the order parameter as 
	$T \to T_c$.
\end{proof}

According to Eqs.~(\ref{UneqZero}), (\ref{DecreasingEaOfTemperature}) (\ref{UVan}), and (\ref{DecreaseF}), we have

\begin{equation}\label{CC2}
	\begin{aligned}
		& U_2(T,B_m(T)) > 0; \\
		& U_2(T,H_{c2}(T)) = 0; \\
		& \frac{\partial U_2(T,B)}{\partial T} < 0, \quad \frac{\partial U_2(T,B)}{\partial B} < 0. \\
	\end{aligned}
\end{equation}

\subsection{Expressions of Vortex Activation Energy \( U(T,B) \)}

Based on the properties of vortex activation energy \( U(T,B) \), we can explore the detailed expressions of \( U(T,B) \). Referring to Eq.~(\ref{U1}), we rewrite \( U_2(T,B) \) as

\begin{equation}\label{U2}
	U_2(T,B) = U_0 u(T,B),
\end{equation}
where \( u(T,B) \) is a function that has similar field dependence as the vortex activation energies constructed at a fixed temperature in earlier literature.

\subsubsection{Expressions of $u(T,B)$}

$u(T,B)$ can be determined by theoretical and experimental analysis.
In earlier studies, a number of activation energy models have been proposed \cite{Anderson1962,Beasley1969,Larkin1979,Feigelman1989,Feigelman1991,Zeldov1989,Zeldov1990,Ma2010,Ma2011}, including linear laws \cite{Anderson1962,Ma2010,Ma2011}, logarithmic laws \cite{Zeldov1989,Zeldov1990}, inverse power laws \cite{Feigelman1989,Feigelman1991}, generalized polynomial model \cite{Ma2010,Ma2011}, and others.
Here, we aim to extend the activation energies by introducing an explicit temperature dependence, thereby ensuring their validity across different temperatures.
To accommodate these considerations, we require that the activation energy \( U(T,B) \) in Eq.~(\ref{U1}) maintains the same field dependence at a fixed \( T \) as those proposed in earlier studies.

For simplicity, we define
\begin{align}\label{BBf}
	\frac{B}{H_{c2}} = \frac{B}{H_{c2}(0) \left( 1 - T / T_c \right)} = \frac{b}{1 - t},
\end{align}
where \( b = B / H_{c2}(0) \) is the scaled magnetic field and \( t = T / T_c \) is the scaled temperature. Referring to the activation energies proposed in earlier studies and the constraint conditions (\ref{CC2}), we write explicit expressions for \( u(T,B) \) as follows:

\begin{itemize}
	\item Inverse-Power Model \cite{Feigelman1989,Feigelman1991}:
	
	\begin{equation}\label{InversePU}
		u(T,B) = \frac{1}{\mu} \left[ \left( \frac{H_{c2}}{B} \right)^\mu - 1 \right] = \frac{1}{\mu} \left[ \left( \frac{1 - t}{b} \right)^\mu - 1 \right].
	\end{equation}
	
	\item Logarithmic Model \cite{Zeldov1989,Zeldov1990}:
	
	\begin{equation}\label{LogU}
		u(T,B) = \ln \left( \frac{H_{c2}}{B} \right) = \ln \left( \frac{1 - t}{b} \right).
	\end{equation}
	
	\item Generalized Polynomial Model \cite{Ma2010,Ma2011}:
	
	\begin{equation}\label{PolyU}
		u(T,B) = \sum_{l=1}^\infty a_l \left( 1 - \frac{B}{H_{c2}} \right)^l = \sum_{l=1}^\infty a_l \left( 1 - \frac{b}{1 - t} \right)^l,
	\end{equation}
	where the coefficients satisfy \( \sum_{l=1}^\infty a_l = 1 \). Here, we did not write out the linear and quadratic activation energy models because they are special cases of Eq.~(\ref{PolyU}).
\end{itemize}

\subsubsection{Expressions of $U_0$}

We can show that \( U_0 \) must be a constant:
\begin{enumerate}
	\item \( U_0 \) is independent of field \( B \):
	
	Because the field (\( B \)) dependence of the activation energies proposed in earlier studies has been confirmed by experimental observations, they should capture all the \( B \) dependence and must not include any other implicit factor of \( B \).
	Otherwise, the mathematical equations cannot be fitted to the experimental data obtained through magnetic field scanning.
	This should also apply to the function \( u(T,B) \). Therefore, \( U_0 \) in Eq.~(\ref{U2}) is independent of \( B \).
	
	\item \( U_0 \) is independent of temperature \( T \):
	
	Since \( U_0 \) is independent of \( B \), it can be either a pure function of \( T \) or a constant.
	If \( U_0 \) were a pure function of \( T \), then this function should have already been absorbed into the temperature factor \( p(T) \) in Eq.~(\ref{TempF}).
	This means that \( U_0 \) is independent of \( T \).
\end{enumerate}

Since \( U_0 \) is independent of field $B$ and is also independent of temperature $T$, it must be a constant.
In fact, the constant \( U_0 \) represents the ``strength'' (or the order of magnitude) of the pinning centers.
Strong pinning centers are usually normal defects and therefore have high \( U_0 \). Weak pinning centers may be variations in the order parameter \( \psi \) (e.g., fractional changes in \( \psi \) caused by strains) and therefore have low \( U_0 \).

\subsubsection{Final Expressions of \( U(T,B) \)}

Finally, we obtain the exact expressions of the activation energies \( U(T,B) \) by referring to Eq.(\ref{U1}) and Eq.(\ref{U2}):
\begin{equation}\label{UFinal}
	U(T, B) = U_0 p(T) u(T, B),
\end{equation}
where $U_0$ is a constant, $p(T)$ is given by Eq.(\ref{TempF}), and $u(T, B)$ is given by Eqs.~(\ref{InversePU})--(\ref{PolyU}), respectively.

\subsection{Applying the Critical Hopping Rate Criterion}

One of the important tasks in studying vortex lattices is to determine the temperature dependence of the melting field \( B_m(T) \).
To achieve this, we need to use exact expressions of the activation energies Eq.(\ref{UFinal}) and critical hopping rate criterion proposed in Proposition \ref{PropositionCriticalHoppingRate}.

Since vortex lattices in superconductors are usually studied at low temperatures where the activation energy \( U(T,B) \gg k_B T \), we will use the approximation in Eq.~(\ref{ReducedFractionOfAtoms}) to represent the fraction of vortices \( f(T, B) \) with energy \( \ge U(T,B) \):
\[
f(T, B) \approx e^{-U(T,B) / k_B T}.
\]

In the magnetic phase diagram, the melting field \( B_m(T) \) is defined as the maximum magnetic field under which the vortex lattice remains stable. Therefore, as the magnetic field increases to the melting field \( B_m(T) \), the hopping rate reaches the critical value:
\begin{equation}\label{CriticalHoppingRateVortexLattice}
	\gamma = f(T, B_m(T)) = e^{-U(T,B_m(T)) / k_B T}.
\end{equation}

By substituting \( U(t,b) \) into Eq.~(\ref{CriticalHoppingRateVortexLattice}), we can obtain the melting field \( B_m(t) \).

\begin{itemize}
	\item Using the Inverse-Power Model Eq.~(\ref{InversePU}):
	
	\begin{equation}\label{InversePBm}
		B_m(t) = H_{c2}(0) \left( 1 - t \right) \left[ 1 + \mu q(t) \right]^{-1/\mu},
	\end{equation}
	where
	\begin{equation}\label{qFunc}
		q(t) = \frac{k_B T_c \ln \gamma^{-1}}{U_0} \frac{t}{\left( 1 - t^2 \right)^2}.
	\end{equation}
	
	\item Using the Logarithmic Model Eq.~(\ref{LogU}):
	
	\begin{equation}\label{LogBm}
		B_m(t) = H_{c2}(0) \left( 1 - t \right) e^{- q(t)}.
	\end{equation}
	
	\item Using the Polynomial Model Eq.~(\ref{PolyU}):
	
	\begin{equation}\label{PolyBm}
		B_m(t) = H_{c2}(0) \left( 1 - t \right) \left[ 1 - \sum_{l=1}^\infty b_l q^l(t) \right],
	\end{equation}
	where the coefficients \( b_l \) are
	\begin{equation}\label{bns}
		\begin{aligned}
			b_l = & \frac{1}{l a_1^l} \frac{1}{l} \sum_{s,p,u,\dots} (-1)^{s+p+u+\cdots} \cdot \\
			& \frac{l(l+1)\cdots(l - 1 + s + p + u + \cdots)}{s! p! u! \cdots} \cdot \\
			& \left( \frac{a_2}{a_1} \right)^s \left( \frac{a_3}{a_1} \right)^p \left( \frac{a_4}{a_1} \right)^u \cdots, \\
		\end{aligned}
	\end{equation}
	where \( s + 2p + 3u + \cdots = l - 1 \). The inverse coefficients \( a_l \) can be obtained by exchanging \( b_l \leftrightarrow a_l \).
\end{itemize}

\section{Numerical calculation}

The literature indicates that equilibrium vacancy concentrations at melting points \cite{Kraftmakher1998} are in the range \( 10^{-4} \) to \( 10^{-3} \).
We now use the experimental data of elements Lead and Aluminum \cite{Feder1958}, and vortex lattice \cite{Ramshaw2012} to evaluate the critical hopping rate parameter \( \gamma \) and check consistency with our model.

\subsection{Lead Pb}

Experimental data of lead \cite{Rumble2024,Kittel2005}:
\begin{itemize}
	\item Activation energies \( E_a \): 0.53 eV
	
	\item Melting temperature $T_m$: 600.61 K
	
	\item Debye temperature $\Theta_D$: 105 K
\end{itemize}

\subsubsection{Classical lattices}

First, we use classical lattice model Eq.(\ref{CriticalHoppingRateParameterClassical}) to calculate the temperature dependence of fraction \( f(T) \) and the critical hopping rate parameter $\gamma_{\text{Pb}}$ for lead.

\begin{itemize}
	\item Calculate \( \epsilon_{\text{Pb}} \):
	
	\[
	\epsilon_{\text{Pb}} = \dfrac{ 0.53 \text{ eV} }{ (8.617333 \times 10^{ -5 } \text{ eV/K}) \times 600.61 \text{ K} } \approx 10.24,
	\]
	and
	\[
	e^{ - \epsilon_{\text{Pb}} } \approx 3.59 \times 10^{ -5 }.
	\]
	
	\item Calculate \( \gamma_{\text{Pb}} \):
	
	Substituting into Eq.(\ref{CriticalHoppingRateParameterClassical}), we have
	
	\[
	\gamma_{\text{Pb}} = 3.59 \times 10^{ -5 } \times \left(1 + 10.24 + \dfrac{ (10.24)^2 }{ 2 }\right) \approx 0.00229.
	\]
\end{itemize}

\subsubsection{Quantum lattice}

Second, we use quantum lattice \cite{Kittel2005,Pathria2021} model Eq.(\ref{FractionOfAtomsQuantum}) to calculate the temperature dependence of fraction \( f(T) \) and the critical hopping rate parameter $\gamma_{\text{Pb}}$ for lead.

\begin{itemize}
	\item Compute \( N_a \):
	
	Convert \( E_a \) to joules:
	\[
	E_a = 0.53\,\text{eV} \times 1.60218 \times 10^{-19}\,\text{J/eV} = 8.49155 \times 10^{-20}\,\text{J}.
	\]
	
	Using the Debye temperature of lead (\( \Theta_D = 105\,\text{K} \)):
	\[
	\hbar \omega_{\text{Pb}} = k_B \Theta_D = (1.380649 \times 10^{-23}\,\text{J/K})(105\,\text{K}) = 1.44968 \times 10^{-21}\,\text{J}.
	\]
	
	Compute \( N_a \):
	\[
	N_a = \left\lceil \frac{E_a}{\hbar \omega_{\text{Pb}}} - \frac{3}{2} \right\rceil = \left\lceil \frac{8.49155 \times 10^{-20}\,\text{J}}{1.44968 \times 10^{-21}\,\text{J}} - 1.5 \right\rceil = \left\lceil 58.591 - 1.5 \right\rceil = 58.
	\]

	\item Compute exponential terms:
	
	\[
	\beta = \frac{1}{k_B T_m} = \frac{1}{(1.380649 \times 10^{-23}\,\text{J/K})(600.61\,\text{K})} = 1.19914 \times 10^{20}\,\text{J}^{-1}.
	\]
	
	\[
	\beta E_a = (1.19914 \times 10^{20}\,\text{J}^{-1})(8.49155 \times 10^{-20}\,\text{J}) = 10.19.
	\]
	
	$$
	\beta \hbar \omega_{\text{Pb}} = \beta k_B \Theta_D = \frac{\Theta_D}{T_m} = \frac{105\,\text{K}}{600.61\,\text{K}} = 0.17487
	$$
	
	\[
	\beta \hbar \omega_{\text{Pb}} N_a = \beta E_a - \frac{3}{2} \beta \hbar \omega_{\text{Pb}} = 10.19 - 0.2623 = 9.928.
	\]
	
	\begin{itemize}		
		\item \( e^{- \beta \hbar \omega_{\text{Pb}} } = e^{- 0.17487 } = 0.8397 \)
		\item \( e^{- \beta \hbar \omega_{\text{Pb}} N_a } = e^{- 9.928 } = 4.87 \times 10^{-5} \)
		\item \( 1 - e^{- \beta \hbar \omega_{\text{Pb}} } = 0.1603 \)
		\item \( \left( 1 - e^{- \beta \hbar \omega_{\text{Pb}} } \right)^2 = 0.0257 \)
		\item \( e^{- 2 \beta \hbar \omega_{\text{Pb}} } = 0.7051 \).
	\end{itemize}

	\item Compute numerator terms:
	
	First term:	
	\[
	\left( N_a^2 + 3 N_a + 2 \right) \left( 1 - e^{- \beta \hbar \omega_{\text{Pb}} } \right)^2 = (58^2 + 3 \times 58 + 2) \times 0.0257 = 90.66.
	\]
	
	Second term:	
	\[
	\left( 2 N_a + 3 \right) e^{- \beta \hbar \omega_{\text{Pb}} } \left( 1 - e^{- \beta \hbar \omega_{\text{Pb}} } \right) = 119 \times 0.8397 \times 0.1603 = 16.01.
	\]
	
	Third term:	
	\[
	2 e^{- 2 \beta \hbar \omega_{\text{Pb}} } = 1.4102.
	\]
	
	\item Compute \( \gamma_{\text{Pb}} \):
	
	\[
	\gamma_{\text{Pb}} = \frac{ e^{- \beta \hbar \omega_{\text{Pb}} N_a } }{2} \times \text{Numerator} = \frac{4.87 \times 10^{-5}}{2} \times \left(90.66 + 16.01 + 1.4102\right) = 0.00264.
	\]
\end{itemize}

\subsubsection{Comparison with experimental data}

Using the classical lattice model, the calculated \( \gamma \) values for lead is approximately \( 0.00229 \).
Using the quantum lattice model, the calculated \( \gamma \) value for lead is \( 0.00264 \).
This shows that the numeric results fall within the experimental range of equilibrium vacancy concentrations at the melting point \( T_m \) reported in the literature (\( 10^{-4} \) to \( 10^{-3} \)).\cite{Feder1958,Kraftmakher1998}

\begin{figure}[htb]
	\centering
	\includegraphics[height=0.6\textwidth]{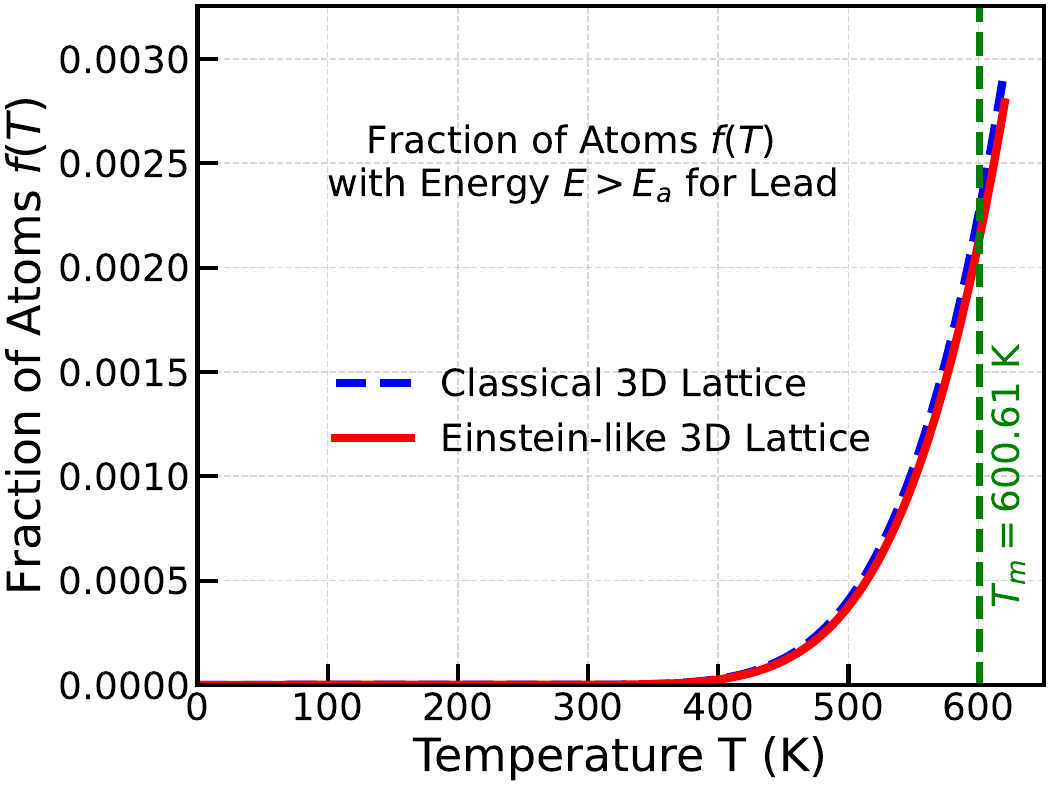}
	\caption{(Color online)
		Fraction of atoms \( f(T) \) with energy \( E \) greater than the activation energy \( E_a \) for Lead Pb.
		The parameters are: activation energies $ E_a = 0.53$ eV, melting temperature $T_m = 600.61$ K, and Debye temperature $\Theta_D = 105$ K.
	}
	\label{FractionLead}
\end{figure}

Fig. \ref{FractionLead} shows the temperature dependence of fraction of atoms \( f(T) \) with energy \( E \) greater than the activation energy \( E_a \) for Lead Pb.
The melting temperature $T_m = 600.61$ K are also specified.
It suggests that our classical lattice model and quantum lattice model predict critical hopping rate parameter for lead \( \gamma_{\text{Pb}} \) consistent with experimental observations for lead.

\subsection{Aluminum}

Experimental data of aluminum \cite{Rumble2024,Kittel2005}:
\begin{itemize}
	\item Activation energies \( E_a \): 0.77 eV
	
	\item Melting temperature $T_m$: 933.47 K
	
	\item Debye temperature $\Theta_D$: 428 K
\end{itemize}

\subsubsection{Classical lattice}

First, we use classical lattice model Eq.(\ref{CriticalHoppingRateParameterClassical}) to calculate the temperature dependence of fraction \( f(T) \) and the critical hopping rate parameter $\gamma_{\text{Al}}$ for aluminum.

\begin{itemize}
	\item Calculate \( \epsilon_{\text{Al}} \):
	
	\[
	\epsilon_{\text{Al}} = \dfrac{ 0.77 \text{ eV} }{ (8.617333 \times 10^{ -5 } \text{ eV/K}) \times 933.47 \text{ K} } \approx 9.57,
	\]
	and
	\[
	e^{ - \epsilon_{\text{Al}} } \approx 6.99 \times 10^{ -5 }.
	\]
	
	\item Calculate \( \gamma_{\text{Al}} \):
	
	Substituting into Eq.(\ref{CriticalHoppingRateParameterClassical}), we have
	\[
	\gamma_{\text{Al}} = 6.99 \times 10^{ -5 } \times \left(1 + 9.57 + \dfrac{ (9.57)^2 }{ 2 }\right) \approx 0.00394.
	\]
\end{itemize}

\subsubsection{quantum lattices}

Second, we use quantum lattice model Eq.(\ref{FractionOfAtomsQuantum}) to calculate the temperature dependence of fraction \( f(T) \) and the critical hopping rate parameter $\gamma_{\text{Al}}$ for aluminum.

\begin{itemize}
	\item Compute \( N_a \):
	
	Convert \( E_a \) to joules:
	
	\[
	E_a = 0.77\,\text{eV} \times 1.60218 \times 10^{-19}\,\text{J/eV} = 1.23368 \times 10^{-19}\,\text{J}.
	\]
	
	Using the Debye temperature for aluminum (\( \Theta_D = 428\,\text{K} \)):
	
	\[
	\hbar \omega_{\text{Al}} = k_B \Theta_D = (1.380649 \times 10^{-23}\,\text{J/K})(428\,\text{K}) = 5.91038 \times 10^{-21}\,\text{J}.
	\]
	
	Compute \( N_a \):
	
	\[
	N_a = \left\lceil \frac{E_a}{\hbar \omega_{\text{Al}}} - \frac{3}{2} \right\rceil = \left\lceil \frac{1.23368 \times 10^{-19}\,\text{J}}{5.91038 \times 10^{-21}\,\text{J}} - 1.5 \right\rceil = \left\lceil 208.69 - 1.5 \right\rceil = 208.
	\]

	\item Compute exponential terms:
	
	\[
	\beta = \frac{1}{k_B T_m} = \frac{1}{(1.380649 \times 10^{-23}\,\text{J/K})(933.47\,\text{K})} = 7.82312 \times 10^{19}\,\text{J}^{-1}.
	\]

	\[
	\beta E_a = (7.82312 \times 10^{19}\,\text{J}^{-1})(1.23368 \times 10^{-19}\,\text{J}) = 9.646.
	\]

	\[
	\beta \hbar \omega_{\text{Al}} = \beta k_B \Theta_D = \frac{\Theta_D}{T_m} = \frac{428\,\text{K}}{933.47\,\text{K}} = 0.45859.
	\]

	\[
	\beta \hbar \omega_{\text{Al}} N_a = \beta E_a - \frac{3}{2} \beta \hbar \omega_{\text{Al}} = 9.646 - 0.6879 = 8.958.
	\]
	
	\begin{itemize}
		\item \( e^{- \beta \hbar \omega_{\text{Al}} } = e^{- 0.45859 } = 0.6320 \)
		\item \( e^{- \beta \hbar \omega_{\text{Al}} N_a } = e^{- 8.958 } = 1.29 \times 10^{-4} \)
		\item \( 1 - e^{- \beta \hbar \omega_{\text{Al}} } = 0.3680 \)
		\item \( \left( 1 - e^{- \beta \hbar \omega_{\text{Al}} } \right)^2 = 0.1354 \)
		\item \( e^{- 2 \beta \hbar \omega_{\text{Al}} } = 0.3994 \).
	\end{itemize}

	\item Compute numerator terms:
	
	First term:	
	\[
	\left( N_a^2 + 3 N_a + 2 \right) \left( 1 - e^{- \beta \hbar \omega_{\text{Al}} } \right)^2 = (208^2 + 3 \times 208 + 2) \times 0.1354 = 59544.
	\]
	
	Second term:	
	\[
	\left( 2 N_a + 3 \right) e^{- \beta \hbar \omega_{\text{Al}} } \left( 1 - e^{- \beta \hbar \omega_{\text{Al}} } \right) = 419 \times 0.6320 \times 0.3680 = 97.35.
	\]
	
	Third term:	
	\[
	2 e^{- 2 \beta \hbar \omega_{\text{Al}} } = 0.7989.
	\]
	
	\item Compute \( \gamma_{\text{Al}} \):
	
	\[
	\gamma_{\text{Al}} = \frac{ e^{- \beta \hbar \omega_{\text{Al}} N_a } }{2} \times \text{Numerator} = \frac{1.29 \times 10^{-4}}{2} \times \left(59544 + 97.35 + 0.7989\right) = 0.00384.
	\]
\end{itemize}

\subsubsection{Comparison with experimental data}

Using the classical lattice model, the calculated \( \gamma_{\text{Al}} \) value for aluminum is approximately \( 0.00394 \).
Using the quantum lattice model, the calculated \( \gamma_{\text{Al}} \) value for aluminum is \( 0.00384 \).
This shows that the numeric results fall within the experimental range of equilibrium vacancy concentrations at the melting point \( T_m \) reported in the literature (\( 10^{-4} \) to \( 10^{-3} \)).\cite{Feder1958,Kraftmakher1998}
It suggests that our classical lattice model and quantum lattice model predict critical hopping rate parameter for aluminum \( \gamma_{\text{Al}} \) consistent with experimental observations for aluminum.

\begin{figure}[htb]
	\centering
	\includegraphics[height=0.6\textwidth]{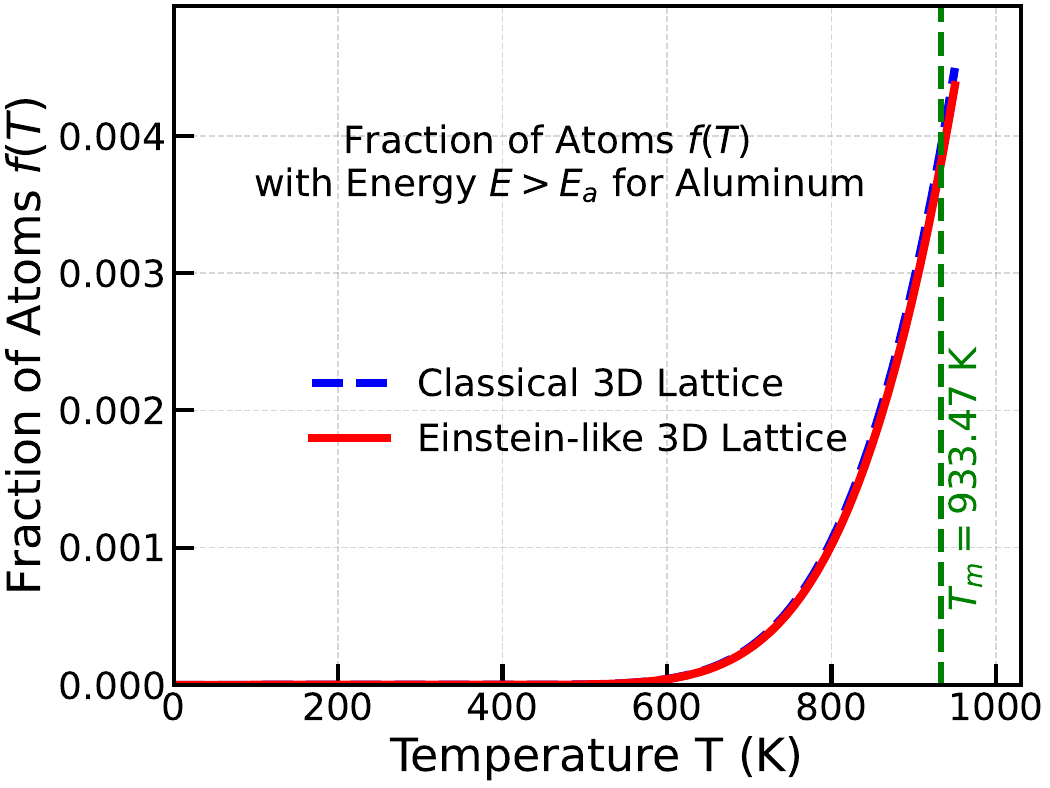}
	\caption{(Color online)
		Fraction of atoms \( f(T) \) with energy \( E \) greater than the activation energy \( E_a \) for aluminum.
		The parameters are: activation energies $ E_a = 0.77$ eV, melting temperature $T_m = 933.47$ K, and Debye temperature $\Theta_D = 428$ K.
	}
	\label{FractionAluminum}
\end{figure}

Fig. \ref{FractionAluminum} shows that the calculated values of critical hopping rate parameters $ \gamma = f(T_m) $ for both lead and aluminum are consistent with experimental data on equilibrium vacancy concentrations at their respective melting points \cite{Feder1958,Kraftmakher1998}.
This suggests that our classical lattice model and quantum lattice model, when using appropriate \( E_a \) values from experimental data, can produce physically meaningful results.

\subsection{Vortex Lattice}

For vortex lattice, we use Eq.~(\ref{InversePBm}) as an example to fit the experimental data points adapted from Ref. \cite{Ramshaw2012}.

\begin{figure}[htb]
	\centering
	\includegraphics[height=0.6\textwidth]{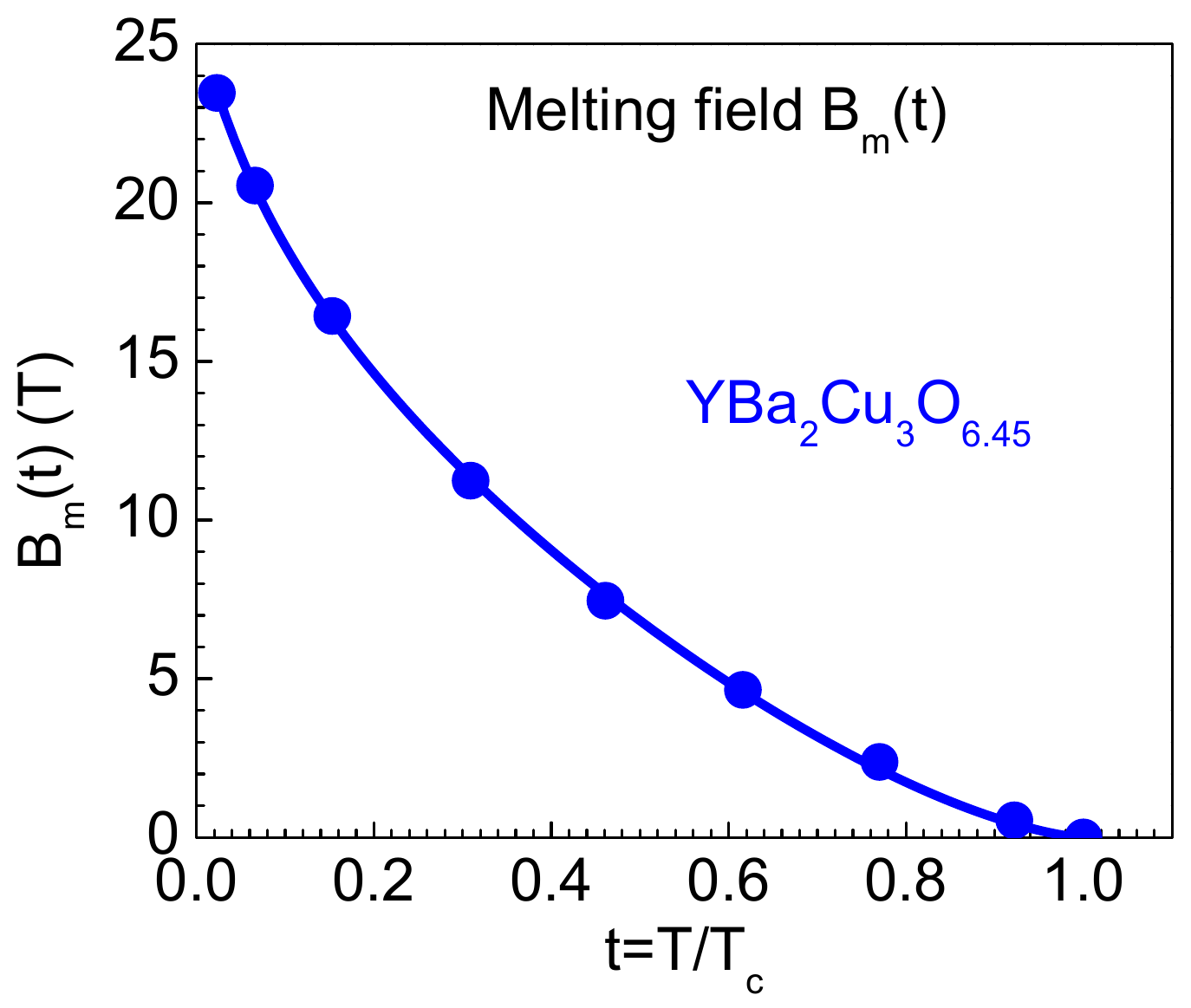}
	\caption{(Color online)
		The melting fields of a $YBa_2Cu_3O_{6.67}$ superconductor.
		The scattering points are the experimental data of melting fields and the solid lines are the theoretical fits to the experimental data with Eq.(\ref{InversePBm}).
		The data points are adapted from Ref.~\cite{Ramshaw2012} and the fitting results are: $H_{c2}(0)=25.8209\pm 0.33785$ T, $k T_c \text{ln} P_m^{-1}/U_0 = 3.54094 \pm 0.50666$, $\mu=4.15171 \pm 0.28223$.
	}
	\label{FigureVortexlattice}
\end{figure}

Fig. \ref{FigureVortexlattice} shows the fits of Eq.~(\ref{InversePBm}), the melting field equation derived with the inverse-power model, to the experimental data from \( YBa_2Cu_3O_{6.67} \) superconductors.

\section{Discussion}

The critical hopping-rate perspective offers several advantages over amplitude-based
criteria. First, it ties directly to thermally activated processes that generate
\emph{defects} (vacancies, interstitials, or local topological disruptions) in the
lattice. While Lindemann's rule captures an average amplitude threshold, it does not
distinguish wave-like phonons from localized barrier-crossing events. Our approach
makes that distinction explicit.

Additionally, we find that near melting, the fraction of particles (or vortices) that
surmount energy barriers typically remains small (\(\sim 10^{-4}\) to \(10^{-3}\)).
Yet even this modest subset is enough to undermine crystalline order if such hopping
occurs at a sufficiently high rate, explaining why the onset of melting is relatively
sharp and strongly influenced by exponential Arrhenius factors.

A key strength of our framework is its versatility. For classical solids, one simply
inserts the relevant activation energy (\(E_a\)) - accounting for anharmonic
effects - into an Arrhenius form to obtain the melting temperature. For vortex lattices,
the same approach employs the vortex pinning energy \(U(T,B)\) to derive the melting
field \(B_m(T)\). In quantum models such as the fermionic or bosonic Hubbard model, the
``activation barrier'' corresponds to on-site repulsion \(U\), and the hopping amplitude
\(t\) plays the role of \(\Gamma\). All these systems thus share a single criterion
for losing order once the ratio \(\Gamma/\nu_0\) or \(t/U\) exceeds a critical
value \(\gamma \approx 10^{-3}\text{--}10^{-1}\).

Nevertheless, certain limitations remain. Like other melting criteria (e.g., Lindemann
or Born), our approach does not emerge \emph{strictly} from first principles within the
harmonic approximation; it is imposed as a physically motivated measure for when barrier
crossings become frequent enough to destroy long-range order. More detailed modeling of
collective defects or pinning centers in inhomogeneous materials may provide further
refinements.

While our primary results focus on classical/Einstein-like lattices and vortex lattices,
the same hopping-based rationale can be extended to quantum lattice models such as the
fermionic or bosonic Hubbard model. We explore these extensions at the mean-field level
in Appendices~\ref{MeltingCriterionFermionicLattices} and
\ref{MeltingCriterionBosonicLattices}. However, we stress that the large discrepancies
with experiment in strongly correlated systems are artifacts of mean-field theory.

Achieving full quantitative agreement likely requires more sophisticated approaches
beyond the mean-field approximation, such as advanced renormalization-group techniques
or numerical simulations. Additionally, future experiments using ultracold atoms in
optical lattices could more directly test the barrier-crossing perspective. For instance,
single-site detection in cold-atom experiments may enable the counting of local hopping
events.

\textbf{Relation to equilibrium thermodynamics.} Our approach does \emph{not} purport
to override or contradict the usual free-energy criterion for melting. In standard
thermodynamics, a first-order melting transition occurs where the free energies of the
solid and liquid phases coincide. We simply propose that \emph{equilibrium} defect
formation - quantified by a hopping rate \(\Gamma \propto \exp(-E_a/k_B T)\) - provides a
practical microscopic marker for this same transition. The phrase ``hopping fraction''
could be misread as purely kinetic or nonequilibrium, but in fact our defect concentration
and activation energies all stem from the equilibrium partition function. Hence, there is
nothing ``deeply unphysical'' about referencing a critical \(\Gamma\). Rather, \(\Gamma/\nu_0\)
is just another way to track how many particles reside in barrier-crossing states at
equilibrium. We find that once this fraction \(\Gamma/\nu_0\) (or \(f(T)\)) passes a small
threshold \(\gamma\), the lattice loses mechanical stability in the same sense described
by Born's criterion or Lindemann's rule.

\section{Conclusion}

We have introduced and tested a melting criterion based on the notion of a
\emph{critical hopping rate} of particles. Our central hypothesis is that
lattice melting is governed by localized hopping events (rather than by large-amplitude
lattice vibrations alone), and that melting ensues once these barrier-crossing
processes occur frequently enough to disrupt long-range order. By combining
a probabilistic analysis of Arrhenius-like rates with the assumption of a
\textit{critical fraction} of particles crossing their barriers, we obtain a
unified criterion that not only reproduces the classical Lindemann and Born
results but also generalizes naturally to superconducting vortex matter and
quantum lattice systems.

This dynamical viewpoint underscores the role of defects or excitations that
break the lattice integrity, providing a deeper connection between collective
phenomena and the onset of disorder. Our comparison with experimental vacancy
data in metals and vortex-melting measurements in superconductors corroborates
the physical plausibility of this criterion. Extensions to ultracold atoms in
optical lattices and to correlated-electron systems suggest that critical
hopping rates (or tunnelings) offer a wide-ranging conceptual tool for
understanding when and why an ordered state ``melts''.

Furthermore, this ``critical hopping-rate'' perspective can be viewed as a restatement
of the conventional melting transition from a defect-based vantage point. By
expressing the melting point in terms of the \emph{equilibrium} fraction of
particles whose energies exceed a barrier \(E_a\), we recover established
melting criteria (e.g., Lindemann's amplitude limit, Born's shear-instability
condition) within a unified framework. Our numerical estimates of this fraction
match well with measured vacancy concentrations in metals and vortex activation
energies in superconductors.

Of course, \emph{non-equilibrium} effects (such as superheating, undercooling,
or rapid quenching) can shift or delay the observed melting temperature from
the ideal equilibrium value. Nevertheless, as a simple and intuitive
\emph{equilibrium-based rule-of-thumb}, the critical hopping fraction provides
a fresh window on the same underlying physics that free-energy arguments
describe. We hope that future studies - especially in vortex lattices or strongly
correlated lattices (Hubbard models) - will find this defect-centered,
dynamical viewpoint a valuable complement to more traditional analyses.

\appendix

\section{Melting Criterion for Fermionic Lattices}
\label{MeltingCriterionFermionicLattices}

In this section, we develop a theoretical framework of quantum version melting criterion for a fermionic lattice based on the critical hopping rate criterion.
We used fermionic Hubbard model \cite{Hubbard1963} and consider the half-filled regime, where the competition between kinetic (hopping) and interaction (on-site Coulomb repulsion) energies leads to a Mott metal-insulator transition \cite{Imada1998}.

Here when we talk about a fermionic lattice ``melting'', we must distinguish between the physical lattice of ions and the electronic state of the system.
In a Mott insulator with antiferromagnetic order, the ions remain in a solid crystalline arrangement - the lattice in the literal sense is not melting into a liquid state.
Instead, what ``melts'' is the long-range magnetic order of the electrons.

We aim to determine the critical hopping amplitude \( t_c(T) \) at which antiferromagnetic order in a fermionic lattice vanishes, giving way to a metallic (paramagnetic) phase.
In essence, as the hopping amplitude \(t\) increases relative to the interaction energy, the lattice undergoes a transition from a Mott insulating state with antiferromagnetic order to a disordered metallic state, effectively ``melting'' its long-range magnetic order.

\subsection{Fermionic Hubbard Model}

The fermionic Hubbard model is a cornerstone of condensed matter theory, describing interacting electrons on a lattice. Its Hamiltonian is \cite{Hubbard1963}:
\begin{equation}\label{HamiltonianFermionicHubbardModel}
	\hat{H} = - t \sum_{\langle i, j \rangle, \sigma} \left( \hat{c}_{i\sigma}^\dagger \hat{c}_{j\sigma} + \text{h.c.} \right) + U \sum_{i} \hat{n}_{i\uparrow} \hat{n}_{i\downarrow} - \mu \sum_{i,\sigma} \hat{n}_{i\sigma},
\end{equation}
\( t \) is the hopping amplitude between nearest-neighbor sites \( \langle i, j \rangle \),
\( U \) is the on-site Coulomb repulsion energy,
\( \mu \) is the chemical potential,
\( \hat{c}_{i\sigma}^\dagger \) and \( \hat{c}_{i\sigma} \) are creation and annihilation operators for fermions with spin \( \sigma = \uparrow, \downarrow \) at site \( i \),
\( \hat{n}_{i\sigma} = \hat{c}_{i\sigma}^\dagger \hat{c}_{i\sigma} \) is the number operator for fermions with spin \( \sigma \) at site \( i \),
``h.c.'' stands for the Hermitian conjugate.

The first term represents kinetic energy due to fermions hopping between nearest-neighbor sites,
the second term represents the on-site Coulomb repulsion when two fermions with opposite spins occupy the same site,
and the last term represents the chemical potential controlling the fermion density.

\subsection{Mean-Field Approximation}

At half-filling (\(\langle \hat{n}_i \rangle = 1\)), the Hubbard model can exhibit a transition from a paramagnetic metal to an antiferromagnetic Mott insulator. To capture the antiferromagnetic order, we introduce a staggered magnetization and apply a mean-field decoupling to the interaction term.\cite{Georges1996}

\subsubsection{Magnetic Order Parameter}

Define the local magnetization:
\begin{equation}\label{DefMagnetizationAtSite}
	m_i = \frac{1}{2} \left( \langle \hat{n}_{i\uparrow} \rangle - \langle \hat{n}_{i\downarrow} \rangle \right).
\end{equation}

Assuming a bipartite lattice (e.g., a square lattice), we consider an antiferromagnetic order where the magnetization alternates between sublattices \( A \) and \( B \):
\[
m_i = \begin{cases}
	m, & i \in A, \\
	- m, & i \in B.
\end{cases}
\]

\subsubsection{Mean-Field Decoupling}

To proceed, we need to perform a mean-field decoupling of the interaction term \cite{Moriya1985}.
Write 
\[
\hat{n}_{i\sigma} = \langle \hat{n}_{i\sigma}\rangle + \delta \hat{n}_{i\sigma}
\]

Assuming that the fluctuations are small, the interaction term \(\hat{n}_{i\uparrow}\hat{n}_{i\downarrow}\) can be approximated by neglecting second-order fluctuations $ \delta \hat{n}_{i\uparrow} \delta \hat{n}_{i\downarrow} $:
\begin{equation}\label{DecoupledInteractionTerm}
	\hat{n}_{i\uparrow}\hat{n}_{i\downarrow} \approx \hat{n}_{i\uparrow}\langle \hat{n}_{i\downarrow} \rangle + \langle \hat{n}_{i\uparrow}\rangle \hat{n}_{i\downarrow} - \langle \hat{n}_{i\uparrow}\rangle\langle \hat{n}_{i\downarrow}\rangle.
\end{equation}

Defining the total density per site
\begin{equation}\label{TotalParticleDensityPerSite}
	n = \langle \hat{n}_{i\uparrow} \rangle + \langle \hat{n}_{i\downarrow} \rangle.
\end{equation}
we have:
\[
\langle \hat{n}_{i\uparrow}\rangle = \frac{n}{2} + m_i, \quad \langle \hat{n}_{i\downarrow}\rangle = \frac{n}{2} - m_i.
\]

Substitute into Eq.~(\ref{DecoupledInteractionTerm}):
\[
\hat{n}_{i\uparrow}\hat{n}_{i\downarrow} \approx \frac{n}{2}\hat{n}_i - m_i(\hat{n}_{i\uparrow} - \hat{n}_{i\downarrow}) - \left(\frac{n}{2}\right)^2 + m_i^2.
\]

This approximation linearizes the interaction term by replacing one of the operators in each product with its average value, effectively reducing the two-body interaction to single-particle terms.

\subsubsection{Mean-Field Hamiltonian}

The mean-field Hamiltonian becomes:
\[
\hat{H}_{\text{MF}} = 
- t \sum_{\langle i, j \rangle,\sigma} ( \hat{c}_{i\sigma}^\dagger \hat{c}_{j\sigma} + \text{h.c.} ) 
+ U \sum_i \left[ \frac{n}{2}\hat{n}_i - m_i(\hat{n}_{i\uparrow}-\hat{n}_{i\downarrow}) \right] 
- \mu \sum_{i,\sigma}\hat{n}_{i\sigma} 
- N U \left( \frac{n}{2} \right )^2 
+ N U m^2,
\]

Defining an effective chemical potential \(\mu_{\text{eff}} = \mu - U\frac{n}{2}\), and using \(\sigma = +1\) for spin-up ($ \uparrow $) and \(\sigma = -1\) for spin-down ($ \downarrow $), we get:
\begin{equation}\label{HamiltonianFermionicHubbardModelMF}
	\hat{H}_{\text{MF}} = 
	- t \sum_{\langle i,j\rangle,\sigma} ( \hat{c}_{i\sigma}^\dagger \hat{c}_{j\sigma} + \text{h.c.} ) 
	- \mu_{\text{eff}} \sum_{i,\sigma}\hat{n}_{i\sigma} 
	- U\sum_{i,\sigma} m_i \sigma \hat{n}_{i\sigma} 
	- N U \left( \frac{n}{2} \right )^2 
	+ N U m^2.
\end{equation}

\subsubsection{Momentum-Space Representation}

We now transform the Hamiltonian to momentum space using the Fourier transform:
\[
\hat{c}_{i\sigma} = \frac{1}{\sqrt{N}}\sum_{\mathbf{k}} e^{i\mathbf{k}\cdot \mathbf{R}_i}\hat{c}_{\mathbf{k}\sigma},
\]
where $ N $ is the number of lattice sites and $ \mathbf{R}_i $ is the position vector of site $ i $.

The kinetic (hopping) term becomes:
\[
- t \sum_{\langle i,j\rangle,\sigma} (\hat{c}_{i\sigma}^\dagger \hat{c}_{j\sigma}+\text{h.c.}) = \sum_{\mathbf{k},\sigma}\epsilon_{\mathbf{k}}\hat{c}_{\mathbf{k}\sigma}^\dagger \hat{c}_{\mathbf{k}\sigma},
\]
where
\begin{equation}\label{EnergyDispersionRelation}
	\epsilon_{\mathbf{k}} = - t \sum_{\delta} e^{ i \mathbf{k} \cdot \delta }
\end{equation}
is the energy dispersion relation and $ \delta $ is the vectors connecting a site to its nearest neighbors.

The on-site terms remain local and transform as:
$$
- \mu_{\text{eff}} \sum_{i, \sigma} \hat{n}_{i \sigma} = - \mu_{\text{eff}} \sum_{\mathbf{k}, \sigma} \hat{c}_{\mathbf{k} \sigma}^\dagger \hat{c}_{\mathbf{k} \sigma}.
$$

Similarly, the magnetization term becomes:
$$
- U \sum_{i, \sigma} m_i \sigma \hat{n}_{i \sigma} = - U \sum_{i, \sigma} m_i \sigma \hat{c}_{i \sigma}^\dagger \hat{c}_{i \sigma}.
$$

We define the spinor:
\begin{equation}\label{DefSpinor}
	\Psi_{\mathbf{k}}^\dagger = \left( \hat{c}_{\mathbf{k} \uparrow}^\dagger, \hat{c}_{\mathbf{k} + \mathbf{Q} \downarrow}^\dagger \right ),
\end{equation}
where \( \mathbf{Q} = (\pi/a, \pi/a) \) is the antiferromagnetic wave vector, such that $ e^{ i \mathbf{Q} \cdot \mathbf{R}_i } = +1 $ for $ i \in A $ and $ -1 $ for $ i \in B $.

The mean-field Hamiltonian in momentum space decouples into blocks. For a bipartite lattice, \(\epsilon_{\mathbf{k}+\mathbf{Q}} = -\epsilon_{\mathbf{k}}\).
The mean-field Hamiltonian in momentum space can be written as:
\[
\hat{H}_{\text{MF}} = \sum_{\mathbf{k}} \Psi_{\mathbf{k}}^\dagger \hat{H}_{\mathbf{k}} \Psi_{\mathbf{k}} + \text{const},
\]
where:
\[
\hat{H}_{\mathbf{k}} = 
\begin{pmatrix}
	\epsilon_{\mathbf{k}}-\mu+Um & 0 \\
	0 & -\epsilon_{\mathbf{k}}-\mu - Um
\end{pmatrix}.
\]
is the single-particle Hamiltonian matrix. Here we dropped the ``eff'' subscript of $\mu_{\text{eff}}$ for simplicity.

Its eigenvalues are:
\begin{equation}\label{EnergyEigenvalues}
	E_{\mathbf{k}}^{\pm} = \pm E_{\mathbf{k}}, \quad E_{\mathbf{k}} = \sqrt{(\epsilon_{\mathbf{k}}-\mu)^2+(Um)^2}.
\end{equation}

Let $E_0 = |\epsilon_{\mathbf{k}} - \mu|$. Then for small $(Um / E_0)^2$,
we can expand the energy eigenvalues Eq.(\ref{EnergyEigenvalues}) using the binomial (or Taylor) expansion for \(\sqrt{1+x}\approx 1 + x/2 - \cdots\):
\begin{equation}\label{EnergyEigenvaluesExpansion1}
	E_{\mathbf{k}} = E_0 \sqrt{1 + \left(\frac{Um}{E_0}\right)^2}
	\approx E_0 \left[ 1 + \frac{1}{2} \left(\frac{Um}{E_0}\right)^2 \right]
	= E_0 + \frac{(U m)^2}{2E_0}.
\end{equation}

\subsubsection{Free energy (grand potential)}

The dispersion splits into two branches $ E_{\mathbf{k}}^+ $ and $ E_{\mathbf{k}}^- $. The grand partition function is:
$$
\mathcal{Z}_{\text{GC}} 
= \mathrm{Tr}\, e^{-\beta(\hat{H}-\mu \hat{N})}
= \prod_{\mathbf{k}}\left[(1+e^{-\beta E_{\mathbf{k}}^+})(1+e^{-\beta E_{\mathbf{k}}^-})\right]
$$

From the Fermi-Dirac distribution function:
\[
f(E) = \frac{1}{1+e^{\beta E}} \Longrightarrow e^{\beta E} = \frac{1- f(E)}{f(E)}.
\]

Therefore,
$$
\mathcal{Z}_{\text{GC}} 
= \prod_{\mathbf{k}}\left[(1+e^{-\beta E_{\mathbf{k}}^+})(1+e^{-\beta E_{\mathbf{k}}^-})\right]
= \prod_{\mathbf{k}} \frac{1}{f(E_k)[1-f(E_k)]}
$$

The mean-field approximation also adds a term $N U m^2$ (up to a constant shift independent of $ m $) in the free energy from the decoupling of interaction (see Eq.(\ref{HamiltonianFermionicHubbardModelMF})).
Finally, the free energy $\Omega(m)$ at temperature $ T = 1/(k_B \beta) $ is given by:
\begin{equation}\label{FreeEnergy0}
	\Omega(m) = -\frac{1}{\beta}\ln \mathcal{Z}_{\text{GC}} = \frac{1}{\beta} \sum_{\mathbf{k}} \ln \left\{ f(E_k)[1-f(E_k)]\right\} + N U m^2.
\end{equation}

Since we are interested in the behavior right at the transition, we focus on the regime where $m$ is very small.
Express the free energy as a function of the magnetization $ m $, expand it near $ m=0 $, and identify the condition under which a nonzero $ m $ can appear spontaneously.

First, let $G(E) = \ln \left\{ f(E)[1-f(E)]\right\}$ and perform a Taylor expansion at $E_0$ with respective to \(\delta = \frac{(U m)^2}{2 E_0}\) (see Eq.(\ref{EnergyEigenvaluesExpansion1})):
\[
G(E_0+\delta) = G(E_0) + \delta G'(E_0) + \frac{\delta^2}{2} G''(E_0) + \cdots
\]

The first derivative:
\[
G'(E) = - \beta \left(1 - 2 f(E)\right).
\]
where we used equation
$$
\frac{d f(E)}{d E} = -\beta f(E)[1 - f(E)].
$$

Therefore,
\[
G(E_0+\delta) \approx G(E_0) - \frac{(U m)^2}{2 E_0} \beta \left[1 - 2 f(E_0)\right].
\]

Substituting into Eq.(\ref{FreeEnergy0}), we have:
\begin{equation}\label{FreeEnergy1}
	\Omega(m) = \frac{1}{\beta} \sum_{\mathbf{k}} G(E_0) + m^2 \left\{ NU - \sum_{\mathbf{k}} \frac{U^2}{2 E_0} \left[1 - 2 f(E_0)\right] \right\} + \cdots
\end{equation}

\subsubsection{Landau Free Energy Expansion and Critical Hopping Amplitude}

On the other hand, according to the Landau theory, we can expand the free energy in powers of $ m $ as:
\begin{equation}\label{FreeEnergyLandauExpansion}
	\Omega(m) = \Omega(0) + \frac{1}{2}A m^2 + \frac{1}{4}B m^4 + \cdots
\end{equation}

According to the critical hopping rate criterion, the fermionic lattice melting occurs when the hopping amplitude \(t\) reaches a critical value \(t_c\), beyond which the system becomes thermodynamically unstable to form a nonzero order parameter.

The system becomes unstable towards developing a nonzero $ m $ if the quadratic coefficient $ A $ changes sign from positive to negative.
The instability occurs at $A = 0$.
Comparing Eq.(\ref{FreeEnergy1})Eq.(\ref{FreeEnergyLandauExpansion}), we have
\begin{equation}\label{QuadraticCoefficientAtCriticalPoint}
	A = NU - \sum_{\mathbf{k}} \frac{U^2}{2 E_0} \left[1 - 2 f(E_0)\right] = 0
\end{equation}

Rearrange this, we obtain the condition for critical hopping amplitude $t_c(T)$:
\begin{equation}\label{CriticalConditionFortc1}
	\boxed{\frac{2}{U} = \frac{1}{N } \sum_{\mathbf{k}} \frac{1 - 2 f(E_0)}{E_0} = \frac{1}{N} \sum_{\mathbf{k}} \frac{1 - 2 f(|\epsilon_{\mathbf{k}}-\mu|)}{|\epsilon_{\mathbf{k}}-\mu|}}.
\end{equation}

\subsubsection{Density of States and Scaling}

Define the density of states (DOS):
\[
D(\epsilon) = \frac{1}{N}\sum_{\mathbf{k}}\delta(\epsilon-\epsilon_{\mathbf{k}}).
\]

Because \(\epsilon_{\mathbf{k}}\) scales linearly with \( t \), we write \(\epsilon = t x\), leading to:
\[
D(\epsilon) = \frac{1}{t}D_0\left(\frac{\epsilon}{t}\right),
\]
where \( D_0(x) \) is the DOS for a scaled energy (already unrelated to $t$).

In the thermodynamic limit, replace the sum over \(\mathbf{k}\) by an integral over energies:
\[
\frac{1}{N}\sum_{\mathbf{k}}(\cdots) \rightarrow \int_{-W}^{W} d\epsilon\, D(\epsilon) (\cdots),
\]
where \(W\) is the half-bandwidth.

Substituting into the critical condition Eq(\ref{CriticalConditionFortc1}):
\[
\frac{2}{U} = \int_{-W}^{W}d\epsilon\,D(\epsilon)\frac{1 - 2f(|\epsilon|)}{|\epsilon|}.
\]

In dimensionless variables \( x = \epsilon/t \):
\begin{equation}\label{CriticalConditionFortc2}
	\boxed{\frac{t_c}{U} = \int_0^{W/t_c} dx\,\frac{D_0(x)}{x} [1 - 2 f(x t_c)]}.
\end{equation}

This implicit equation determines \( t_c(T) \). The solution for \( t_c \) usually requires numerical methods for a given density of states.

\subsubsection{Melting Temperature $T_m$ in the Strong-Coupling Limit}

When \(U \gg t\), the Hubbard model at half-filling maps onto a Heisenberg antiferromagnet \cite{Anderson1959,Manousakis1991,Sandvik2010} with exchange \(J \approx \tfrac{4t^2}{U}\).
The Néel temperature \(T_m\) (the temperature at which AF order vanishes) scales roughly with \(J\). Thus,
\[
k_B\,T_m \;\sim\; \alpha\,J \;=\; \alpha \,\frac{4\,t^2}{U},
\]
so
\[
T_m \;\approx\; \frac{4\,\alpha\,t^2}{U\,k_B}.
\]
In three-dimensional bipartite lattices, \(\alpha\approx 0.95\) is often quoted based on numerical studies of the Heisenberg model.
This relation highlights the interplay between the kinetic scale (\(\sim t\)) and the interaction scale (\(\sim U\)): increasing \(U\) lowers \(T_m\), reflecting that strong onsite repulsion suppresses kinetic delocalization.

The melting temperature \( T_m \) is defined by the point where the antiferromagnetic order vanishes at the critical hopping \( t_c(T_m) \). In the strong-coupling limit (\( U \gg t \)), the Hubbard model at half-filling maps onto the Heisenberg model with exchange coupling \( J = 4 t^2/U \).

The Néel temperature \( T_m \) scales with \( J \):
\[
k_B T_m \approx \alpha J = \alpha \frac{4 t_c^2}{U}.
\]
Thus:
\[
T_m = \frac{4\alpha t_c^2}{U k_B}.
\]

For a 3D cubic lattice, \(\alpha \approx 0.946\), giving:
\[
T_m \approx \frac{3.78 t_c^2}{U k_B}.
\]

This relation shows that the melting temperature decreases with increasing \( U \) and increases with \( t_c^2 \), reflecting the interplay between kinetic and interaction energy scales.

\subsection{Renormalization Group (RG) Corrections}

The mean-field (MF) analysis in the preceding sections gives a ``bare'' (unrenormalized) estimate for the critical hopping amplitude \( t_c(T) \) and the melting temperature \( T_m \).
In the mean-field approach, we replaced interaction terms by their average values (mean fields). While this yields a good first approximation - particularly in higher dimensions - critical points are strongly influenced by fluctuations of the order parameter near the transition.

However, near criticality, fluctuations of the order parameter can alter these values.
Renormalization group (RG) methods offer a systematic way to include these fluctuation effects. 
The RG procedure ``integrates out'' fast (short-wavelength) fluctuations step by step, rescaling the system at each stage.
This process yields flow equations for the effective coupling constants. Solving these equations (even approximately) shows how the parameters \((t, U, T)\) are renormalized as we approach the critical point.
Below, we sketch how to implement RG corrections to the mean-field results.

\subsubsection{Landau-Ginzburg-Wilson (LGW) Free Energy}

In practice, implementing RG requires identifying an effective Landau-Ginzburg-Wilson (LGW) free energy functional \cite{Chaikin2000,Goldenfeld2018,ShankarRG1994} for the antiferromagnetic (AF) order parameter.
We then integrate out fluctuations in momentum shells, see how the ``mass'' (which determines whether the system is ordered or disordered) flows, and find the corrected condition for criticality.
This procedure leads to shifts in \( t_c(T) \) and \( T_m \) compared to their mean-field values.

Near the AF transition, the order parameter \(m(\mathbf{r})\) is a slowly varying field (the staggered magnetization).
We write a coarse-grained free energy functional:
\[
\mathcal{F}[m] \;=\; \int d^d r \,\Biggl[
\tfrac12 A_0\,m^2(\mathbf{r})
\;+\; \tfrac12 Z_0 \bigl(\nabla m(\mathbf{r})\bigr)^2
\;+\; \tfrac{1}{4!} B_0\,m^4(\mathbf{r})
\;+\;\dots
\Biggr],
\]
where
\(A_0\) is the ``bare mass'' parameter: in mean-field theory, \(A_0 \propto t - t_c(T)\). It vanishes at the mean-field critical point,
\(Z_0\) is the ``bare stiffness'' (gradient term),
and \(B_0\) is a quartic coupling that controls fluctuations.

These coefficients - often called ``bare'' parameters - can be extracted from the microscopic Hubbard model once we identify the AF order parameter.
For instance, from the mean-field analysis, the quadratic coefficient $A$ vanishes at the mean-field critical line \(t = t_c^{(0)}(T)\) (see Eq.(\ref{QuadraticCoefficientAtCriticalPoint})).

\subsubsection{RG Procedure: Integrating Out Fast Modes}

The Wilsonian RG approach proceeds by:
\begin{enumerate}
	\item Defining a momentum shell. We integrate out Fourier components of \(m(\mathbf{r})\) with wavevectors \(\Lambda/b < |\mathbf{k}| < \Lambda\), where \(\Lambda\) is an ultraviolet (UV) cutoff (on the order of the inverse lattice spacing), and \(b>1\) is a rescaling factor.
	
	\item Renormalizing parameters. This integration changes \(A_0, Z_0, B_0\) into new ``effective'' couplings \(A_0^\prime, Z_0^\prime, B_0^\prime\).
	
	\item Rescaling. We then rescale lengths \(\mathbf{r}\to b\,\mathbf{r}\) and fields \(m(\mathbf{r})\to b^{(d+2-\eta)/2}\,m\bigl(\mathbf{r}\bigr)\), where \(\eta\) is the anomalous dimension. The result is a new free energy of the same form but with updated coefficients. Repeating this process iteratively produces ``flow equations'' for the couplings as functions of the RG scale \(l = \ln b\).
\end{enumerate}

Conceptually, at the critical point, the mass parameter \(A(l)\) flows to zero as \(l\to\infty\).
This ``fixed point'' condition yields the renormalized critical line in terms of the original (bare) parameters \((t, U, T)\).

\subsubsection{One-Loop RG Flow Equations}

To one-loop order in a \(\phi^4\)-type theory (the standard LGW theory for an Ising-like order parameter in \(d\) dimensions), the flow equations near the Gaussian fixed point often take the form:
\[
\frac{d A}{d l} \;=\; \bigl(2 - \eta\bigr)\,A \;-\; C_A\,B\,I_A,  
\quad
\frac{d B}{d l} \;=\; \bigl(4 - d - \eta\bigr)\,B \;-\; C_B\,B^2\,I_B,
\quad
\frac{d Z}{d l} \;=\; \eta\,Z.
\]
where
\(A, B, Z\) are the renormalized mass, coupling, and stiffness, respectively,
\(I_A, I_B\) are one-loop momentum-shell integrals (they depend on \(\Lambda\) and on the running mass \(A\)),
and \(C_A, C_B\) are numerical constants depending on dimensionality.

The anomalous dimension \(\eta\) is determined self-consistently from the fixed-point conditions. In many common scenarios (e.g., near \(d=4\)), \(\eta\) can be small, and we set \(\eta\approx 0\) at one-loop. While the full derivation of these flow equations can be lengthy, the key takeaway is that \(A\) flows to zero at criticality, giving the location of the phase boundary when fluctuations are included.

\subsubsection{Corrections to the Critical Hopping Amplitude $t_c(T)$}

Referring to the critical condition Eq.(\ref{CriticalConditionFortc2}), we include RG corrections:
\begin{enumerate}
	\item The on-site repulsion \(U\) and the hopping \(t\) are replaced by their effective (running) values \(U_{\mathrm{eff}}(l), t_{\mathrm{eff}}(l)\) at a large scale \(l^*\).
	
	\item The DOS \(D_0(x)\) is replaced by a slightly renormalized DOS \(D_{\mathrm{eff}}(x;l)\) (since the band structure and the effective dispersion can be modified by fluctuations).
	
	\item The Fermi function \(f(\dots)\) also depends on the renormalized temperature \(T_{\mathrm{eff}}(l)\).
\end{enumerate}

Hence, the fluctuation-corrected critical condition takes a schematic form:
\[
\frac{t_{\mathrm{eff}}(l^*)}{\,U_{\mathrm{eff}}(l^*)\,} 
\;=\; 
\int_0^{W/t_{\mathrm{eff}}(l^*)} 
dx\,\frac{D_{\mathrm{eff}}(x;l^*)}{x}\,\bigl[1 - 2\,f\bigl(x\,t_{\mathrm{eff}}(l^*)\bigr)\bigr].
\]
Solving this equation for \(\,t_{\mathrm{eff}}(l^*)\) gives the renormalized critical hopping amplitude:
\[
t_c(T) 
\;=\;
t_c^{(0)}(T) \;+\; \delta t_c(T),
\]
where
\[
\delta t_c(T)
\;\approx\;
t_c^{(0)}(T)\,\Bigl[\alpha_U\bigl(U,t_c^{(0)},T\bigr)\;-\;\alpha_t\bigl(U,t_c^{(0)},T\bigr)\Bigr],
\]
and the functions \(\alpha_U,\alpha_t\) are integrals stemming from the one-loop RG diagrams.
The net effect: at finite dimension \(d\), the RG typically modifies the mean-field boundary by a multiplicative correction (often a few percent to tens of percent, depending on dimension and distance from the upper critical dimension).

\section{Melting Criterion for Bosonic Lattices}
\label{MeltingCriterionBosonicLattices}

In this section, we study the melting transition of a bosonic lattice using the Bose-Hubbard model first introduced in \cite{Gersch1963}.
This model has been extensively analyzed in the context of superfluid-Mott transitions \cite{Fisher1989,Elstner1999,Kuhner1998}, and realized experimentally in ultracold atomic gases \cite{Jaksch1998,Greiner2002,Bloch2008}.

The critical hopping rate $ t_c $ is then the parameter at which the system transitions from a Mott insulator to a superfluid phase.
Traditional mean-field theory provides initial insights but often lacks precision near critical points due to the neglect of fluctuations.
By incorporating renormalization group (RG) techniques, we can account for these fluctuations and improve the accuracy of $ t_c $ and the melting temperature $ T_m $.

We will derive the expressions of critical hopping rate \( t_c(T) \) for the Bose-Hubbard model at both finite temperature ($ T > 0 $) and zero temperature ($ T = 0 $).
We'll also derive the expression for the melting temperature $ T_m $.

\subsection{Bose-Hubbard Model}

The Bose-Hubbard model describes interacting bosons on a lattice and is given by the Hamiltonian:
\begin{equation}\label{BoseHubbardHamiltonian0}
	\hat{H} = -t \sum_{\langle i, j \rangle} ( \hat{a}_i^\dagger \hat{a}_j + \hat{a}_j^\dagger \hat{a}_i ) + \frac{U}{2} \sum_i \hat{n}_i ( \hat{n}_i - 1 ) - \mu \sum_i \hat{n}_i,
\end{equation}
where
\( t \) is the hopping amplitude between nearest-neighbor sites \( \langle i, j \rangle \),
\( U \) is the on-site interaction strength,
\( \mu \) is the chemical potential,
\( \hat{a}_i^\dagger \) and \( \hat{a}_i \) are bosonic creation and annihilation operators at site \( i \),
and \( \hat{n}_i = \hat{a}_i^\dagger \hat{a}_i \) is the number operator at site \( i \).

\subsection{Mean-Field Approximation}

\subsubsection{Superfluid Order Parameter \( \psi \)}

We use mean-field approximation to simplify the calculation.
Define the superfluid order parameter \( \psi \) as the expectation value of the annihilation operator:
\[
\psi = \langle \hat{a}_i \rangle.
\]

$\psi$ measures the coherent superfluid order in the system. We can choose $\psi$ to be real.
In the Mott insulator, $\psi = 0$ because particles are localized, and there is no phase coherence.
A non-zero $\psi$ indicates delocalization and long-range phase coherence, which are hallmarks of the superfluid phase.

\subsubsection{Mean-field decoupling}

The mean-field decoupling is a technique used to simplify interacting many-body systems by approximating interaction terms with effective single-particle terms.
Let us first decouple the hopping term: rewrite the hopping term in a way that is easier to handle, especially when calculating expectation values or deriving self-consistency equations.

Assuming a uniform system, \( \psi \) is the same at every site.
Express operators as fluctuations around the mean: write the annihilation and creation operators as the sum of their mean value and fluctuations:
\[
\hat{a}_i = \psi + (\hat{a}_i - \psi), \quad \hat{a}_i^\dagger = \psi^* + (\hat{a}_i^\dagger - \psi^*).
\]
where \( \hat{a}_i - \psi \) represents the fluctuation at site \( i \) and has a zero expectation value: $\langle \hat{a}_i - \psi \rangle = 0$.

Consider the product \( \hat{a}_i^\dagger \hat{a}_j \):
\[
\hat{a}_i^\dagger \hat{a}_j = (\psi^* + \delta \hat{a}_i^\dagger)(\psi + \delta \hat{a}_j) = \psi^* \psi + \psi^* \delta \hat{a}_j + \delta \hat{a}_i^\dagger \psi + \delta \hat{a}_i^\dagger \delta \hat{a}_j,
\]
where \( \delta \hat{a}_i^\dagger = \hat{a}_i^\dagger - \psi^* \) and \( \delta \hat{a}_j = \hat{a}_j - \psi \).

\( \psi^* \psi \) is a constant term representing the product of mean fields.
\( \psi^* \delta \hat{a}_j \) and \( \delta \hat{a}_i^\dagger \psi \) are terms linear in fluctuations.
\( \delta \hat{a}_i^\dagger \delta \hat{a}_j \) is a term quadratic in fluctuations.
In the mean-field approximation, fluctuations are considered small. We neglect the quadratic term \( \delta \hat{a}_i^\dagger \delta \hat{a}_j \), as it involves the product of two small quantities.
The approximation becomes:
\begin{align}
	\begin{aligned}
		\hat{a}_i^\dagger \hat{a}_j 
		&\approx \psi^* \psi + \psi^* (\hat{a}_j - \psi) + (\hat{a}_i^\dagger - \psi^*) \psi \\
		&= \psi^* \psi + \psi^* \hat{a}_j - \psi^* \psi + \hat{a}_i^\dagger \psi - \psi^* \psi \\
		&= \hat{a}_i^\dagger \psi + \psi^* \hat{a}_j - |\psi|^2,
	\end{aligned}
\end{align}

\( \hat{a}_i^\dagger \psi \) and \( \psi^* \hat{a}_j \) represent the interaction of an individual particle with the average field \( \psi \).
\( |\psi|^2 \) is a constant term that can often be absorbed into the overall energy or treated separately.
This approximation transforms a two-operator term into terms involving only single operators, greatly simplifying calculations in the mean-field framework.

\subsubsection{Mean-field Hamiltonian}

Substituting back, the mean-field Hamiltonian becomes:
$$
\hat{H}^{\text{MF}} = \sum_i \left[ -t z (\psi^* \hat{a}_i + \psi \hat{a}_i^\dagger) + \frac{U}{2} \hat{n}_i(\hat{n}_i - 1) - \mu \hat{n}_i + t z |\psi|^2 \right],
$$
where
$ z $ is coordination number.

We focus on the Bose-Hubbard model's mean-field Hamiltonian for a single site $ i $:
\begin{equation}\label{BoseHubbardMeanFieldHamiltonianSingleSite}
	\hat{H}_i^{\text{MF}} = -t z (\psi^* \hat{a}_i + \psi \hat{a}_i^\dagger) + \frac{U}{2} \hat{n}_i(\hat{n}_i - 1) - \mu \hat{n}_i  + t z |\psi|^2
\end{equation}

When diagonalizing the single-site Hamiltonian or calculating the order parameter, the constant term \( t z |\psi|^2 \) can be neglected because it does not influence the minimization process. However, this term still contributes to the system's total energy.

\subsubsection{Perturbative Expansion}

Near the critical point, $ \psi \to 0 $, so we can consider $ \psi $ as a small parameter and linearize the equations: we can treat the hopping term linear in \( \psi \) as perturbation.

\begin{itemize}
	\item Unperturbed Hamiltonian:
	
	\[
	\hat{H}_0 = \frac{U}{2} \hat{n}_i ( \hat{n}_i - 1 ) - \mu \hat{n}_i.
	\]
	whose eigenstates are $ |n\rangle $ with eigenvalues (energies):
	\begin{equation}\label{UnperturbedEnergyEigenvalues}
		E_n^{(0)} = \frac{U}{2} n(n - 1) - \mu n
	\end{equation}
	
	\item Perturbation:
	
	\[
	\hat{V} = - t z ( \psi^* \hat{a}_i + \psi \hat{a}_i^\dagger ) + t z |\psi|^2.
	\]
\end{itemize}

We will compute the expectation value \( \psi = \langle \hat{a}_i \rangle \) up to second order in \( \psi \).

\subsubsection{Critical Hopping Rate \( t_c(T) \)}

Increasing the temperature introduces additional thermal excitations that facilitate particle hopping, effectively lowering the barrier for particles to escape the localized state. At finite temperatures, a competition arises between thermal fluctuations, which promote delocalization, and strong interactions that can still localize particles. Consequently, both the critical hopping amplitude \( t_c \) and the melting temperature \( T_m \) are influenced by a combination of quantum and thermal fluctuations.

In this work, we refer to \(T_m\) as the temperature at which the system transitions from a localized (Mott) phase to a delocalized (superfluid) phase in the Bose-Hubbard model.  Concretely, we define \(T_m\) by the condition that the superfluid order parameter \(\psi\) first becomes nonzero at any finite hopping \(t > 0\).
Equivalently, this corresponds to the boundary in the \(t\text{--}T\) plane separating insulating and superfluid behavior.  
Although we call it a ``melting temperature'' in analogy with classical solids, one should keep in mind that this signals the disappearance of the Mott gap and the emergence of a quantum superfluid, not a classical liquid.

In our case, since we're computing the expectation value in a thermal ensemble, we need to sum over all states, each weighted by its Boltzmann factor \( e^{ - \beta E_n^{(0)} } \).
Define the occupation probabilities:
\begin{equation}\label{OccupationProbabilities}
	P_n = \frac{ e^{ - \beta E_n^{(0)} } }{ Z_0 }.
\end{equation}
where
\begin{equation}\label{UnperturbedPartitionFunction}
	Z_0 = \sum_{n=0}^{\infty} e^{ - \beta E_n^{(0)} }
\end{equation}

\subsubsection{First-Order Correction to \( \psi \)}

The first-order correction to \( \psi \) is:
\[
\psi = \sum_n P_n \langle  n^{(0)} | \hat{a}_i |  n^{(0)} \rangle = 0,
\]
since \( \langle  n^{(0)} | \hat{a}_i |  n^{(0)} \rangle = 0 \), where $|  n^{(0)} \rangle$ is the eigenstate of unperturbed Hamiltonian $\hat{H}_0$.

The first-order correction \( | n^{(1)} \rangle \) to the state \( | n^{(0)} \rangle \) is given by:
\begin{equation}\label{FirstOrderCorrectionToStateN0}
	| n^{(1)} \rangle = \sum_{ m \neq n } \frac{ \langle m^{(0)} | \hat{V} | n^{(0)} \rangle }{ E_n^{(0)} - E_m^{(0)} } | m^{(0)} \rangle.
\end{equation}

\subsubsection{Second-Order Correction to \( \psi \)}

The leading non-zero contribution comes from second-order perturbation theory:

\[
\psi = \sum_{ n } P_n \left( \langle n^{(0)} | \hat{a}_i | n^{(1)} \rangle + \langle n^{(1)} | \hat{a}_i | n^{(0)} \rangle \right ),
\]

Substituting Eq.(\ref{FirstOrderCorrectionToStateN0}), we obtain the expectation value:
\[
\psi = \sum_{ n } P_n \left( \sum_{ m \neq n } \frac{ \langle n^{(0)} | \hat{a}_i | m^{(0)} \rangle \langle m^{(0)} | \hat{V} | n^{(0)} \rangle }{ E_n^{(0)} - E_m^{(0)} } + \sum_{ m \neq n } \frac{ \langle n^{(0)} | \hat{V} | m^{(0)} \rangle \langle m^{(0)} | \hat{a}_i | n^{(0)} \rangle }{ E_n^{(0)} - E_m^{(0)} } \right ).
\]

However, since \( \hat{V} \) and \( \hat{a}_i \) involve \( \hat{a}_i \) and \( \hat{a}_i^\dagger \), the non-zero contributions are from states where \( m = n \pm 1 \).
Therefore, we have
\[
\psi = \sum_n P_n \left[ \frac{ \langle n | \hat{V} | n + 1 \rangle \langle n + 1 | \hat{a}_i | n \rangle }{ E_n^{(0)} - E_{ n + 1 }^{(0)} } + \frac{ \langle n | \hat{V} | n - 1 \rangle \langle n - 1 | \hat{a}_i | n \rangle }{ E_n^{(0)} - E_{ n - 1 }^{(0)} } + \text{c.c.} \right ].
\]

The relevant matrix elements are:
\begin{itemize}
	\item For \( m = n + 1 \):
	\begin{align}
		\begin{aligned}
			& \langle n | \hat{V} | n + 1 \rangle = - t z \psi^* \sqrt{ n + 1 } \\
			& \langle n + 1 | \hat{a}_i | n \rangle = \sqrt{ n + 1 } \\
			& E_n^{(0)} - E_{ n + 1 }^{(0)} = - ( U n - \mu )
		\end{aligned}
	\end{align}
	
	\item For \( m = n - 1 \):
	\begin{align}
		\begin{aligned}
			& \langle n | \hat{V} | n - 1 \rangle = - t z \psi \sqrt{ n } \\
			& \langle n - 1 | \hat{a}_i | n \rangle = \sqrt{ n } \\
			& E_n^{(0)} - E_{ n - 1 }^{(0)} = U ( n - 1 ) - \mu
		\end{aligned}
	\end{align}	
\end{itemize}

The first term:
\[
\frac{ \langle n | \hat{V} | n + 1 \rangle \langle n + 1 | \hat{a}_i | n \rangle }{ E_n^{(0)} - E_{ n + 1 }^{(0)} } 
= \frac{ - t z \psi^* \sqrt{ n + 1 } \cdot \sqrt{ n + 1 } }{ - ( U n - \mu ) } 
= \frac{ t z \psi^* ( n + 1 ) }{ U n - \mu },
\]

The second term:
\[
\frac{ \langle n | \hat{V} | n - 1 \rangle \langle n - 1 | \hat{a}_i | n \rangle }{ E_n^{(0)} - E_{ n - 1 }^{(0)} } 
= \frac{ - t z \psi \sqrt{ n } \cdot \sqrt{ n } }{ U ( n - 1 ) - \mu } 
= \frac{ - t z \psi n }{ U ( n - 1 ) - \mu },
\]

Combining the two terms and simplifying:
\[
\psi = \sum_n P_n \left[ \frac{ t z \psi^* ( n + 1 ) }{ U n - \mu } - \frac{ t z \psi n }{ U ( n - 1 ) - \mu } + \text{c.c.} \right ].
\]

\( \psi^* \) and \( \psi \) are complex conjugates. Assuming \( \psi \) is real, we can write:
\[
\psi = 2 t z \psi \sum_n P_n \left[ \frac{ n + 1 }{ U n - \mu } - \frac{ n }{ U ( n - 1 ) - \mu } \right].
\]

Bring all terms to one side and factor out \( \psi \), we obtain the following self-consistency equation:
\begin{equation}\label{PhaseTransitionEquation}
	\psi \left\{ 1 - 2 t z \sum_n P_n \left[ \frac{ n + 1 }{ U n - \mu } - \frac{ n }{ U ( n - 1 ) - \mu } \right] \right\} = 0.
\end{equation}

Let us now analyze the physical meaning of Eq.(\ref{PhaseTransitionEquation}).
Applying the critical hopping rate criterion: the bosonic lattice melting occurs when the hopping rate $t$ reaches a critical value $t_c$, beyond which the system becomes dynamically unstable.
Specifically:

\begin{itemize}
	\item Below the Critical Point (\( t < t_c \))
	
	The Mott insulating phase is stable and there is no superfluid phase.
	The only solution is \( \psi = 0 \) because $\psi = \langle \hat{a}_i \rangle$ is superfluid order parameter.

	\item At the Critical Point (\( t = t_c \))
	
	The system is at the onset of instability. The superfluid phase appears and a non-zero \( \psi \ne 0 \) becomes possible. 
	In order for Eq.(\ref{PhaseTransitionEquation}) to hold, the coefficient must be zero:
	\[
	1 - 2 t_c(T) z \sum_n P_n \left[ \frac{ n + 1 }{ U n - \mu } - \frac{ n }{ U ( n - 1 ) - \mu } \right] = 0
	\]
\end{itemize}

Therefore, at the melting temperature $T_m$, the critical hopping rate $t_c(T_m)$ (solution of Eq.(\ref{PhaseTransitionEquation})) is:
\begin{equation}\label{CriticalHoppingRateAtFiniteTemperature}
	\boxed{t_c(T_m)  = \frac{1}{2 z \sum\limits_n P_n \left[ \frac{ n + 1 }{ U n - \mu } - \frac{ n }{ U ( n - 1 ) - \mu } \right]}}.
\end{equation}

\subsubsection{Critical Hopping Rate at \( T = 0 \) K}

Here, thermal fluctuations are absent, and the primary mechanism for bosons to ``hop'' and delocalize is quantum tunneling.
Consequently, the transition from Mott insulator to superfluid at \(T=0\) is a quantum phase transition, driven by the ratio \(t/U\).

For bosons, at \( T = 0 \) K, only the ground state with \( n = n_0 \) is occupied (\( P_{ n_0 } = 1 \)).
The critical hopping rate simplifies to:
\begin{equation}
	t_c(0)  = \frac{1}{2 z \left[ \frac{ n_0 + 1 }{ U n_0 - \mu } - \frac{ n_0 }{ U ( n_0 - 1 ) - \mu } \right]}
\end{equation}

This indicates the existence of superfluid phase at zero temperature.

\subsection{Renormalization Group (RG) Corrections}

Mean-field theory provides an intuitive starting point for estimating the critical hopping amplitude \( t_c \) and melting temperature \( T_m \) in the Bose-Hubbard model.
However, mean-field ignores spatial and temporal fluctuations, which become especially relevant near criticality and in low dimensions.\cite{Elstner1999,Kuhner1998,Sansone2008}
In this subsection, we outline how renormalization group (RG) methods \cite{Donner2007,Scalettar1991} systematically account for these fluctuations and lead to quantitative corrections to \( t_c \) and \( T_m \).

\subsubsection{Effective Action and RG Setup}

A convenient way to study the Mott-superfluid (SF) transition from a field-theoretic perspective is to move from the Bose-Hubbard Hamiltonian to a path-integral or imaginary-time formulation \cite{Sachdev2011}.  After a Hubbard-Stratonovich transformation or a direct expansion in the superfluid order parameter \(\psi(\mathbf{r},\tau)\), one obtains an effective action of the form
\[
S_{\mathrm{eff}}[\psi] \;=\; \int \! d^d r \int_0^\beta d\tau \;\Bigl[\,
\psi^*(\partial_\tau - \mu_\mathrm{eff})\,\psi
\;+\; \kappa\,|\nabla \psi|^2 
\;+\; r_0\,|\psi|^2 
\;+\; g\,|\psi|^4 
\;+\;\cdots
\Bigr],
\]
where
\( \beta = 1/(k_B T) \) is the inverse temperature,
\( \mu_\mathrm{eff} \) is an effective chemical potential that may absorb some constants from the original Hamiltonian,
\( r_0 \) is the tuning parameter whose sign determines whether \(\psi\) is nonzero (SF phase) or zero (Mott phase) at the saddle point,
\( \kappa \) sets the kinetic energy of phase fluctuations (related to the hopping \(t\)),
\( g\) is the quartic self-interaction coupling emerging from bosonic on-site repulsion \(U\).

Near the Mott-SF transition, the coefficient \( r_0 \) vanishes at criticality.  In a more microscopic approach (e.g., strong-coupling expansion near the tip of the Mott lobe), one identifies \( r_0 \propto (t - t_c^{(\mathrm{MF})}) \).  The mean-field solution is found by minimizing \( r_0\,|\psi|^2 + g\,|\psi|^4/2 + \cdots \), ignoring fluctuations around \(\psi(\mathbf{r},\tau)\).  However, large phase (and amplitude) fluctuations strongly modify this picture as one approaches the transition in low dimensions or at finite \(T\).

\subsubsection{Momentum-Shell RG Flow}

One classic RG approach integrates out the high-momentum (short-wavelength) modes of \(\psi(\mathbf{r},\tau)\) shell by shell.  We briefly sketch the steps:
\begin{enumerate}
	\item Decompose \(\psi\) into slow modes \(\psi_{<}\) (momenta \(|\mathbf{k}|\leq \Lambda/b\)) and fast modes \(\psi_{>}\) (momenta \(\Lambda/b < |\mathbf{k}| \leq \Lambda\)), where \(\Lambda\) is an ultraviolet cutoff set by the lattice spacing, and \(b>1\) is the rescaling factor.
	
	\item Integrate out fast modes \(\psi_{>}\) in the partition function:
	\[
	Z \;=\; \int \! \mathcal{D}\psi_{<}\,\mathcal{D}\psi_{>}\;\,
	e^{-S_{\mathrm{eff}}[\psi_{<}+\psi_{>}]} 
	\;\to\; 
	\int \! \mathcal{D}\psi_{<}\;\,e^{-S_{\mathrm{eff}}^{\prime}[\psi_{<}]},
	\]
	which produces an effective action \(S_{\mathrm{eff}}^{\prime}\) solely in terms of \(\psi_{<}\).  
	
	\item Rescale coordinates (\(\mathbf{r}\to \mathbf{r}/b\)) and imaginary time (\(\tau \to \tau/b^z\), where \(z\) is the dynamical critical exponent) to restore the original cutoff \(\Lambda\).  Simultaneously, rescale the field \(\psi\to b^{\Delta}\psi\) to keep the kinetic term in canonical form.
\end{enumerate}

At each RG step, the parameters \((r_0, g, \kappa, \ldots)\) flow according to differential equations of the form
\[
\frac{d r_0}{d \ell} \;=\; \beta_{r}(r_0, g, \ldots), 
\quad
\frac{d g}{d \ell} \;=\; \beta_{g}(r_0, g, \ldots),
\quad
\cdots 
\]
where \(\ell = \ln b\).  The transition occurs where \(r_0\) flows to zero and the system is scale-invariant.
\begin{itemize}
	\item In \(d>2\), the stable critical fixed point often matches the 3D \(XY\) (or \(\phi^4\)) universality class.  The coupling \(g\) may flow to a nontrivial value \(g^*\).
	
	\item In \(d=2\) or \(d=1\), strong fluctuations or topological excitations can significantly shift the transition relative to mean-field predictions.
\end{itemize}

\subsubsection{One-Loop Correction to the Critical Line $t_c(T)$}

A more direct, ``strong-coupling'' RG expansion uses the Bose-Hubbard Hamiltonian itself, focusing on the Mott insulator regime and treating the kinetic term (hopping) as a perturbation.
At the leading (one-loop) order, one often finds that the mean-field critical line \(t_c^\mathrm{(MF)}(T)\) is shifted upward \cite{Sansone2008}, i.e.,
\[
t_c^\mathrm{(RG)}(T)
\;=\;
t_c^\mathrm{(MF)}(T)\,\Bigl[1 + \alpha_1 + \cdots \Bigr],
\]
where \(\alpha_1>0\) depends on:
\begin{itemize}
	\item Dimension \(d\): Stronger infrared fluctuations in lower \(d\) enhance corrections.  
	\item On-site interaction \(U\) and chemical potential \(\mu\): These set the gap in the Mott state and control how easily particles/holes are created by hopping.
	\item Temperature \(T\): Thermal fluctuations can further assist hopping.  Near \(T=0\), quantum fluctuations dominate, while at higher \(T\), the thermal population of excited states modifies the flow of coupling constants.
\end{itemize}

In many references (e.g., Refs. [1,2]), the leading one-loop diagrammatic correction is found by evaluating bubble or tadpole diagrams associated with the boson propagators.  The net effect can be viewed as a renormalization of the effective mass of the bosonic field, shifting the location of the gap closure point.  

Physically, this means that mean-field underestimates how quickly the Mott gap closes when fluctuations are properly accounted for.
One might interpret this as the system needing a slightly larger hopping \(t\) to overcome the residual short-range correlations that pin the particles in place.

\subsubsection{Finite Temperature and the ``Melting'' of the Mott Phase}

At nonzero temperature ($T>0$), the path-integral extends over a finite imaginary time \(\beta\).
The effective dimensionality can shift from \(d+1\) (quantum) to something more classical if \(\beta U \ll 1\). 
Hence, the interplay of quantum and thermal fluctuations modifies the scale at which interactions become relevant or irrelevant.\cite{Fisher1989}

One way to combine RG with the mean-field approach is to replace the mean-field expression for \(t_c(T)\) by
\[
t_c^{(\mathrm{RG})}(T)
\;=\;
t_c^{(\mathrm{MF})}(T)\,
\bigl[1 + \alpha_1\,F(T) + \cdots\bigr],
\]
where \(F(T)\) is a function capturing the temperature dependence of the fluctuation correction.  Some approximate forms (e.g., \(F(T) \approx 1 - e^{-\beta U}\)) capture how the Mott gap is thermally smeared for moderate temperatures.

Recall that we define \(T_m\) as the temperature at which the Mott state is destroyed (the superfluid emerges) at an arbitrarily small \(t>0\).  In an RG language, one would show that the flow of the mass term \(r_0\) changes sign at any \(t>0\) beyond a certain temperature threshold.  Thus, more refined RG treatments show that \(T_m\) is not just a naive mean-field limit: fluctuations can either lower or raise the boundary, depending on dimension and other parameters, typically shifting it from the mean-field curve.

\subsubsection{Quantitative Estimates and Universality Classes}

\begin{itemize}
	\item In 3D:  
	The Mott-SF transition generally belongs to the 3D \(XY\) universality class if one focuses on the particle-hole symmetric point.  RG corrections are relatively smaller (often 5-15\% shift from mean-field in typical 3D optical lattice experiments).
	
	\item In 2D:  
	The transition is in the 2D \(XY\) universality class, with more substantial fluctuations.  Vortex binding/unbinding (Kosterlitz-Thouless type) can also play a role, especially near certain parameter regimes.
	
	\item In 1D:\cite{Giamarchi2004}  
	The system can be described by the Tomonaga-Luttinger liquid formalism in the superfluid phase, and a Bose-Mott transition occurs with even stronger quantum fluctuations.  RG treatments typically yield larger deviations from mean-field \(t_c\).
\end{itemize}

In all cases, quantum Monte Carlo simulations \cite{Scalettar1991,Sansone2008,Sandvik2010} confirm that the RG-corrected critical lines provide an excellent match to numerical data, whereas pure mean-field predictions can differ significantly (especially in low \(d\)).

\section{Hamiltonian of Quasi-harmonic Oscillators in a Crystalline Lattice}
\label{AppendixHamiltonianQuasiHarmonicOscillators}

In a crystalline lattice, atoms are bound in potential wells but are not static; they vibrate about their equilibrium positions \cite{Ashcroft1976,Landau1980,Kittel2005}.
We will demonstrate that near the melting point, atoms can be considered as three-dimensional quasi-harmonic oscillators by incorporating the influences of anharmonicity, phonon interactions, and thermal expansion into the angular frequency of atomic vibrations and the activation energy of hopping motions.

The Hamiltonian \( H \) of a crystalline lattice of \( N \) atoms can be written as

\begin{equation}\label{GeneralHamiltonian}
	H = T + V.
\end{equation}

The kinetic energy \( T \) is given by
\[
T = \sum_{i=1}^{N} \frac{\mathbf{p}_i^2}{2 m_i},
\]
with \( \mathbf{p}_i \) being the momentum and \( m_i \) the mass of the \( i \)-th atom.

The potential energy \( V \) can be expanded as a Taylor series around the equilibrium positions:

\begin{equation}\label{PotentialEnergy1}
	\begin{aligned}
		V &= V_0 + \sum_{i} \left( \frac{\partial V}{\partial \mathbf{u}_i} \bigg|_{\mathbf{u}=0} \mathbf{u}_i \right) + \frac{1}{2} \sum_{i,j} \left( \frac{\partial^2 V}{\partial \mathbf{u}_i \partial \mathbf{u}_j} \bigg|_{\mathbf{u}=0} \mathbf{u}_i \mathbf{u}_j \right) \\
		&\quad + \frac{1}{3!} \sum_{i,j,k} \left( \frac{\partial^3 V}{\partial \mathbf{u}_i \partial \mathbf{u}_j \partial \mathbf{u}_k} \bigg|_{\mathbf{u}=0} \mathbf{u}_i \mathbf{u}_j \mathbf{u}_k \right) + \dots,
	\end{aligned}
\end{equation}
where
\[
\mathbf{u}_i = \mathbf{r}_i - \mathbf{R}_i^0
\]
is the displacement of the \( i \)-th atom from its equilibrium position \( \mathbf{R}_i^0 \), and \( \mathbf{r}_i \) is its actual position.

In Eq.~(\ref{PotentialEnergy1}), the first derivative term vanishes at equilibrium, the second derivative term corresponds to the harmonic potential, and the third and higher-order derivatives introduce anharmonicity.

\subsection{Anharmonicity}

Anharmonic effects modify the effective stiffness of the bonds between atoms as temperature increases. This change in stiffness is reflected in the temperature-dependent force constants \( U_{ij}(T) \). Instead of explicitly including the anharmonic terms, the quasi-harmonic approximation effectively incorporates anharmonicity by making the harmonic force constants temperature-dependent:
\begin{equation}\label{AnharmonicForceConstants}
	U_{ij}(T) = U_{ij}^0 + \Delta U_{ij}(T),
\end{equation}
where \( U_{ij}^0 \) are the harmonic force constants at absolute zero or a reference temperature, and \( \Delta U_{ij}(T) \) captures the effects of anharmonicity as a function of temperature.

By adjusting \( U_{ij}(T) \), the phonon frequencies \( \omega_k(T) \) become temperature-dependent:
\[
\omega_k^2(T) = \frac{1}{m_k} \sum_{i,j} U_{ij}(T) e^{-i \mathbf{k} \cdot (\mathbf{R}_i - \mathbf{R}_j)}.
\]

Recall that the free energy \( F_{\text{vib}}(T) \) depends on \( \omega_k(T) \), and the entropy \( S = - \left( \frac{\partial F_{\text{vib}} }{\partial T} \right)_V \). Therefore, the entropy \( S \) changes due to the temperature dependence of \( U_{ij}(T) \), incorporating phonon-phonon interactions and entropy changes.

\subsection{Thermal Expansion}

As temperature increases, atoms vibrate with larger amplitudes due to increased thermal energy. Anharmonicity causes the potential energy wells to be asymmetric, leading to a net shift in the average positions of atoms (thermal expansion). This effect can be modeled by allowing the equilibrium positions to be temperature-dependent.

Let \( \Delta \mathbf{R}_i(T) \) be the temperature-dependent shift from the original equilibrium position \( \mathbf{R}_i^0 \).
The new equilibrium position becomes
\[
\mathbf{R}_i(T) = \mathbf{R}_i^0 + \Delta \mathbf{R}_i(T).
\]

The modified displacements are then
\[
\mathbf{u}_i = \mathbf{r}_i - \mathbf{R}_i(T).
\]

\subsection{Final Quasi-harmonic Hamiltonian}

Substituting Eq.~(\ref{AnharmonicForceConstants}) and the modified displacements into Eq.~(\ref{GeneralHamiltonian}), we obtain the final quasi-harmonic Hamiltonian:
\begin{equation}\label{QuasiHarmonicHamiltonian}
	H_{\text{QH}} = \sum_{i=1}^N \left( \frac{\mathbf{p}_i^2}{2m_i} \right) + \frac{1}{2} \sum_{i,j=1}^N U_{ij}(T) \mathbf{u}_i \mathbf{u}_j.
\end{equation}

This Hamiltonian serves as the foundation for studying the thermodynamic and vibrational properties of crystalline solids at high temperatures, where anharmonic effects cannot be neglected.

To facilitate further calculations, we simplify Eq.~(\ref{QuasiHarmonicHamiltonian}) by diagonalizing \( U_{ij}(T) \) using normal mode coordinates \( Q_k \) and corresponding frequencies \( \omega_k(T) \):
\[
\mathbf{u}_i = \sum_{k} e_{ik} Q_k,
\]
where \( e_{ik} \) are the components of the eigenvector corresponding to mode \( k \).

Expressing the Hamiltonian in normal mode coordinates, we have
\begin{equation}\label{HamiltonianNormalModeCoordinates}
	H_{\text{QH}} = \sum_{k} \left( \frac{P_k^2}{2m_k} + \frac{1}{2} m_k \omega_k^2(T) Q_k^2 \right),
\end{equation}
where \( P_k \) is the momentum conjugate to \( Q_k \), and \( m_k \) is the effective mass for mode \( k \).

The angular frequency \( \omega_k \) of each oscillator is related to its spring constant \( k_k \) and mass by \( \omega_k = \sqrt{ \frac{k_k}{m_k}} \).
Thus, the spring constant can be expressed as
\[
k_k = m_k \omega_k^2.
\]

Eq.~(\ref{HamiltonianNormalModeCoordinates}) resembles the Hamiltonian of harmonic oscillators; thus, we refer to them as quasi-harmonic oscillators.
Each quasi-harmonic oscillator is decoupled and therefore independent, allowing us to treat each one separately.

\bibliography{references}

\end{document}